\ifdefined\XeTeXrevision
\else
\pdfoutput=1
\fi
\documentclass[11pt]{amsart}

\usepackage[T1]{fontenc}
\usepackage[utf8]{inputenc}
\usepackage[main=english]{babel}
\usepackage{lmodern}
\usepackage[a4paper,margin=1in]{geometry}
\usepackage[final,protrusion=true,expansion=false]{microtype}
\usepackage{amsmath,amssymb,mathtools}
\usepackage{amsthm,mathrsfs}
\usepackage{enumitem}
\usepackage[sort,compress]{cite}

\usepackage[hypertexnames=false]{hyperref}
\hypersetup{%
	colorlinks=true,
	linkcolor=blue,
	citecolor=blue,
	urlcolor=blue,
	pdftitle={Crystallographic Torus Bundles: Obstructions, Elastic Sectors, and Global Acoustic Modes},
	pdfauthor={Aleksey Prots},
	pdfkeywords={crystalline phase, acoustic phonons, torus bundle, nonsymmorphic symmetry, monodromy, connection}
}

\theoremstyle{plain}
\newtheorem{proposition}{Proposition}[section]
\newtheorem{theorem}[proposition]{Theorem}
\newtheorem{corollary}[proposition]{Corollary}
\theoremstyle{definition}
\newtheorem{definition}[proposition]{Definition}
\theoremstyle{remark}
\newtheorem{remark}[proposition]{Remark}

\title[Crystallographic Torus Bundles]{Crystallographic Torus Bundles: Obstructions, Elastic Sectors, and Global Acoustic Modes}
\author{Aleksey Prots}
\address{Department of Theoretical Physics and Computer Technologies, Kuban State University, Krasnodar, Russian Federation}
\email{BWPSM.lab@gmail.com}
\date{August 23, 2026; revised October 3, 2026}

\begin{document}

\begin{abstract}
We formulate the global translational order parameter of a crystal with a
fixed finite orientational background as a section of an affine torus bundle
induced by the crystallographic extension
$1\to\Pi\to\Gamma\to P\to1$. For a principal $P$-bundle $Q$, the
configuration bundle is $Y_\Pi=Q\times_P T_\Pi$. We identify two distinct
global conditions associated with the same extension class: smooth
configurations exist precisely when the pulled-back class
$f_Q^*[c]\in H^2(X;\Pi_Q)$ vanishes, whereas parallel configurations exist
precisely when the restriction of $[c]$ to the monodromy subgroup vanishes.
Thus a global translational configuration may exist even when a parallel
reference configuration is topologically obstructed.

For phase-independent quadratic first-order theories, the affine torsor
structure yields an exact global fluctuation identity and a
background-independent Hessian about every stationary configuration, without
requiring a parallel reference state. We realize the distinction explicitly
on a compact flat three-manifold with half-turn holonomy and a nonsymmorphic
$C_2$-extension. Although no parallel configuration exists, smooth strained
equilibria do. Their homotopy sectors form a torsor for
$\mathbb Z^5\oplus(\mathbb Z/2\mathbb Z)^2$. We determine the minimizers and
minimum energy in every sector, derive a strictly positive lower energy bound
over all configurations of the quadratic model, and compute the complete
isotropic elastic spectrum, including the kernel and first positive
eigenvalue. For fixed macroscopic quotient, flat linear bundle, and
constitutive coefficients, changing the microscopic affine extension can
therefore change the global configuration space and sector energies while
leaving the quadratic fluctuation operator unchanged.
\end{abstract}

\maketitle

\section{Introduction}

Long-wavelength acoustic phonons are described locally by displacement
fields and their derivatives. This paper studies the global configuration
space and the global fluctuation problem when the orientational background
has a finite crystallographic structure group. The translation phase takes
values in $T_\Pi=\mathbb R^d/\Pi$, and its local representatives are glued by
an affine point-group action in the nonsymmorphic case. The configuration
field is a section of $Y_\Pi=Q\times_P T_\Pi$, where $Q$ is fixed.
It is a translational phase, not a full nonlinear deformation map.
The theory concerns the classical continuum acoustic sector; optical
branches, Brillouin-zone band structure and quantization are outside its scope.

\medskip
\noindent\textbf{Context and contribution.}
The local elastic equations and acoustic velocities recovered here are
standard \cite{BornHuang,LandauLifshitzElasticity,NyePhysicalProperties,MusgraveCrystalAcoustics}.
The Goldstone interpretation of phonons and the redundancy of independent
rotational Goldstone variables are established features of effective theories
\cite{LeutwylerPhonons,LowManohar,WatanabeMurayama2013,BraunerInverseHiggs,BraunerBook}.
Related nonlinear realizations and solid effective theories are discussed in
\cite{ColemanWessZumino1969,CallanColemanWessZumino1969,NicolisEtAl2015,HayataHidaka2014}.
We use these results as physical motivation; freezing $Q$ does not itself
derive an inverse-Higgs constraint.
The order-parameter space $\mathrm E(d)/\Gamma$ is classical
\cite{KlemanMichel1978,Mermin1979}, and torus-valued displacement fields have
already been used in variational models of dislocations
\cite{ScalaVanGoethem2019}. General affine torus bundles, their monodromy and
twisted Chern classes are also established constructions
\cite{BaragliaTorusMonodromy}. Equivariant Fourier methods and spectral
multiplicity formulas for differential forms and natural bundles on flat
manifolds are developed in
\cite{MiatelloRossettiForms,MiatelloRossettiTwisted}; scalar half-turn modes
are discussed in \cite{EskiltHalfTurn}. Our isotropic vector calculation
specializes these methods, with the longitudinal and transverse elastic
weights separated by standard Hodge theory.

The geometric defect literature includes descriptions by torsion and curvature
\cite{BilbyBulloughSmith1955,Kroner1981,deWit1973,KatanaevVolovich1992,Katanaev2005,KleinertGaugeFields,MuraMicromechanics};
elasticity--fracton dualities provide a different related framework
\cite{PretkoRadzihovsky2018,PretkoZhaiRadzihovsky2019,PretkoChenYou2020,GromovRadzihovsky2024}.
The present construction does not introduce dynamical defects. A smooth
orientational reduction alone need not describe a locally compatible,
undeformed crystal: its induced flat metric connection can have torsion.
We therefore distinguish the general bundle construction from the compatible
backgrounds on which it recovers ordinary linear elasticity.

The main purpose of this work is to determine how the affine extension data
of a crystallographic group enter the global configuration and fluctuation
problems of continuum elasticity. The individual ingredients---crystallographic
extensions, affine torus bundles, obstruction classes, and spectral theory on
flat manifolds---are established. The new point is their combination into a
single global field-theoretic problem in which the same crystallographic
extension controls two distinct notions of admissibility and, in the model
below, has a direct energetic consequence. In particular:
\begin{enumerate}[label=\textup{(\roman*)}]
\item We identify the affine torus fibration of $\mathrm E(d)/\Gamma$ and the
configuration bundle $Y_\Pi=Q\times_P T_\Pi$, retaining the translational
part of a nonsymmorphic extension rather than only its linear point-group
action.
\item We show that the obstruction to a smooth translational configuration is
$f_Q^*[c]\in H^2(X;\Pi_Q)$, whereas the stronger obstruction to a parallel
configuration is $[c]|_{H_Q}$. When sections exist, their vertical homotopy
classes form a torsor for $H^1(X;\Pi_Q)$.
\item The affine torsor structure gives the exact global identity
$\beta_{\sigma_0+\varphi}=\beta_{\sigma_0}+D^E\varphi$. Consequently, for
phase-independent quadratic first-order theories, the Hessian is global and
independent of the background; stationarity rather than parallelness is the
relevant condition for the linear fluctuation problem.
\item On a compatible compact flat three-manifold we realize
$f_Q^*[c]=0$ but $[c]|_{H_Q}\ne0$: smooth equilibria exist although a
parallel reference is forbidden. We classify all homotopy sectors, determine
their minimizing equilibria and minimum energies, prove a strictly positive
lower energy bound across the full configuration space of the quadratic
model, and compute the complete isotropic vector spectrum.
\end{enumerate}
Thus the central result is not a new general classification of affine torus
bundles, a new jet formalism, or a new local acoustic speed. Rather, it is the
explicit relation between crystallographic extension data, global
configuration obstructions, sectorwise elastic energy, and the global
fluctuation problem. The half-turn example includes explicit sector
representatives, an attainment criterion for the energy bound, the
self-adjoint operator domain, all mode families, the kernel and the first
positive eigenvalue. It describes screw-periodic identifications rather than
a finite Euclidean body with a free surface. The constitutive law is
prescribed at the continuum level rather than derived from a particular
material.

\medskip
\noindent\textbf{Conventions and scope.}
All manifolds and bundles are smooth; bases are paracompact. The spatial
manifold is $M$, of dimension $d$. The field-theory base is $X=M$ in the
static case and $X=\mathbb R_t\times M$ in the dynamical case. Spatial
structures are pulled back along $\mathrm{pr}_M:X\to M$ in the latter case.
The orientational reduction is time independent. For the abstract bundle and
section theorems $Q\to X$ can be any principal $P$-bundle. For the elastic
interpretation it is a reduction of $F_g(M)$ and is additionally required to
be preserved by $\nabla^g$. This forces $g$ to be locally flat but permits
nontrivial finite global holonomy. Integrals use densities, so the general
variational construction does not require an orientation.

A section need not exist for an arbitrary $Q$. Whenever a background is used,
its existence is assumed or proved. A parallel section is a sufficient
zero-gradient reference for the derivative-only quadratic theory; it is not
a synonym for every equilibrium or every zero-strain field. Smooth sections,
parallel sections and stationary sections are separate notions.
On a simply connected base the monodromy is trivial, so nonsymmorphicity
alone does not obstruct a parallel section.

Sections~\ref{sec:crystal_symmetry}--\ref{sec:associated_bundle} develop the
configuration geometry and the two existence criteria.
Section~\ref{sec:first_order_lagrangian} gives the variational theory and its
compatible elastic limit. Section~\ref{sec:halfturn} contains the complete
half-turn example. Appendix~\ref{app:cubic_example} checks cubic elasticity.

\section{Crystallographic Symmetry Data}
\label{sec:crystal_symmetry}

A crystalline phase is characterized by a residual discrete subgroup of the Euclidean symmetry group. 
In an idealized isotropic medium the full Euclidean group acts as a spatial symmetry, whereas in a crystal only the motions compatible with the lattice remain: translations by lattice vectors and finite rotations belonging to the point group. 
Crystallization may therefore be viewed as a standard example of spontaneous breaking of spatial symmetries: continuous translations are reduced to a lattice, and continuous rotations are reduced to a finite subgroup of \(O(d)\) \cite{Chaikin,BornHuang,LeutwylerPhonons,NicolisEtAl2015}.

\medskip

\subsection{Spatial setting}
\label{subsec:spatial_setting}
Let \(E\) be a Euclidean space of dimension \(d\), identified with \(\mathbb{R}^d\) after a choice of origin. 
The residual symmetry of an ideal crystal is described by a discrete subgroup
\[
\Gamma \subset \mathrm{E}(d)=\mathbb{R}^d\rtimes O(d),
\]
where \(\mathrm{E}(d)\) denotes the group of Euclidean motions of \(E\).

\begin{definition}[Crystallographic group]
	\label{def:crystallographic_group}
	A crystallographic group is a discrete subgroup
	\[
	\Gamma\subset \mathrm{E}(d)
	\]
	acting on the Euclidean space \(E\) by isometries such that the orbit space \(E/\Gamma\) is compact. 
	Equivalently, by the Bieberbach theorems, \(\Gamma\) contains a full-rank translation subgroup of finite index \cite{Charlap1986}.
\end{definition}

In particular, the subgroup of pure translations
\begin{equation}
	\Pi:=\Gamma\cap \mathbb{R}^d
	\label{eq:lattice_def}
\end{equation}
is a full-rank lattice, \(\Pi\simeq \mathbb{Z}^d\). 
The quotient
\begin{equation}
	P:=\Gamma/\Pi
	\label{eq:point_group_def}
\end{equation}
is finite and is realized, through the linear parts of the elements of \(\Gamma\), as a subgroup of \(O(d)\). 
It is called the point group of the crystal. 
Thus \(\Gamma\) fits into the short exact sequence
\begin{equation}
	1\longrightarrow \Pi \longrightarrow \Gamma \longrightarrow P \longrightarrow 1.
	\label{eq:cryst_extension}
\end{equation}
In this group-extension formulation, the symmorphic case corresponds to a split extension, so the crystallographic group has the form
\[
\Gamma \simeq \Pi \rtimes P.
\]
If the extension does not split, such a semidirect-product presentation is not available in general. 
This nonsymmorphic case is represented by nontrivial affine translational data; geometrically, it includes screw axes and glide planes \cite{BradleyCracknell,InternationalTablesA}. 
For the general theory of crystallographic and space groups, see, for example, \cite{Charlap1986,BradleyCracknell,InternationalTablesA}.

\subsection{Symmorphic and Nonsymmorphic Cases}
\label{subsec:symmorphic_nonsymmorphic_affine}

The short exact sequence \eqref{eq:cryst_extension} contains more information than the linear action of the point group on the lattice. 
This additional information is encoded by the extension class of \eqref{eq:cryst_extension}, or equivalently by a group \(2\)-cocycle. 
It is this class that distinguishes the symmorphic and nonsymmorphic cases. 
For the construction below, this distinction has to be realized not only algebraically, but also as an action on the translation torus \(T_\Pi=\mathbb R^d/\Pi\). 
We use the standard description of group extensions by cocycles; see, for example, \cite{Charlap1986,BrownCohomology}.

Write an element of the Euclidean group as
\[
(a,A)\in \mathbb R^d\rtimes O(d),
\]
where \(a\in\mathbb R^d\) is the translation part and \(A\in O(d)\) is the orthogonal part. 
It acts on \(x\in\mathbb R^d\) by
\[
x\longmapsto Ax+a,
\]
and the group law is
\[
(a,A)(b,B)=(a+Ab,AB).
\]
The inclusion \(\Pi\subset\Gamma\) identifies each lattice vector \(\lambda\in\Pi\) with the pure translation \((\lambda,I)\). 
For every coset \(p\in P=\Gamma/\Pi\), choose a representative
\[
\gamma_p=(a_p,A_p)\in\Gamma.
\]
The orthogonal part \(A_p\) is independent of the chosen representative, because changing the representative by a lattice translation changes only \(a_p\). 
Moreover, conjugation by \(\gamma_p\) sends lattice translations to lattice translations:
\[
\gamma_p(\lambda,I)\gamma_p^{-1}=(A_p\lambda,I),
\qquad
\lambda\in\Pi.
\]
Hence \(A_p\Pi=\Pi\). 
Thus the point group acts on the lattice by orthogonal automorphisms,
\[
A:P\longrightarrow \mathrm{Aut}(\Pi)\cap O(d),
\qquad
p\longmapsto A_p.
\]

\begin{proposition}[Affine action of the point group on the translation torus]
	\label{prop:affine_point_action}
	Let \(\Gamma\subset \mathbb R^d\rtimes O(d)\) be a crystallographic group. 
	Let
	\[
	\Pi=\Gamma\cap\mathbb R^d
	\]
	be its translation lattice, and let
	\[
	P=\Gamma/\Pi
	\]
	be its point group. 
	For each \(p\in P\), choose a representative
	\[
	\gamma_p=(a_p,A_p)\in\Gamma.
	\]
	Then the following statements hold.
	
	\begin{enumerate}[label=\textup{(\roman*)}]
		\item The map \(p\mapsto A_p\) defines a representation
		\[
		A:P\to \mathrm{Aut}(\Pi)\cap O(d).
		\]
		
		\item The element
		\begin{equation}
			c(p,q):=a_p+A_pa_q-a_{pq}
			\label{eq:extension_cocycle}
		\end{equation}
		belongs to \(\Pi\) and defines a \(2\)-cocycle of the group \(P\) with values in the \(P\)-module \(\Pi_A\), where the action of \(P\) on \(\Pi\) is given by the representation \(A\).
		
		\item Changing the representatives \(\gamma_p\) changes \(c\) by a coboundary. Hence the class
		\[
		[c]\in H^2(P,\Pi_A)
		\]
		is an invariant of the extension \eqref{eq:cryst_extension}. The extension splits, that is \(\Gamma\simeq\Pi\rtimes P\), if and only if this class is zero.
		
		\item The formula
		\begin{equation}
			\rho(p)[v]=[A_pv+a_p],
			\qquad [v]\in T_\Pi=\mathbb R^d/\Pi,
			\label{eq:affine_torus_action}
		\end{equation}
		defines a left affine action of \(P\) on the torus \(T_\Pi\). 
		In the symmorphic case, after a suitable choice of origin and representatives, one has \(a_p=0\), and \(\rho\) becomes the linear action \([v]\mapsto[A_pv]\).
	\end{enumerate}
\end{proposition}

\begin{proof}
	For \(\lambda\in\Pi\), conjugation by \(\gamma_p=(a_p,A_p)\) gives
	\[
	\gamma_p(\lambda,I)\gamma_p^{-1}=(A_p\lambda,I)\in\Pi.
	\]
	Hence \(A_p\Pi=\Pi\). Moreover, replacing \(\gamma_p\) by \((b_p,I)\gamma_p\), \(b_p\in\Pi\), changes only the translation part \(a_p\), not the orthogonal part \(A_p\). Thus \(p\mapsto A_p\) is well defined as a map \(P\to\mathrm{Aut}(\Pi)\cap O(d)\).
	
	For \(p,q\in P\), multiplication in \(\mathbb R^d\rtimes O(d)\) gives
	\[
	\gamma_p\gamma_q
	=
	(a_p+A_pa_q,A_pA_q).
	\]
	Since \(\gamma_p\gamma_q\) and \(\gamma_{pq}\) project to the same element \(pq\in P\), the element \(\gamma_p\gamma_q\gamma_{pq}^{-1}\) lies in \(\Pi\). Its orthogonal part is therefore the identity, so \(A_pA_q=A_{pq}\). Thus \(A:P\to\mathrm{Aut}(\Pi)\) is a representation. Its translation part gives
	\[
	c(p,q):=a_p+A_pa_q-a_{pq}\in\Pi.
	\]
	Associativity of multiplication in \(\Gamma\) gives
	\[
	A_pc(q,r)-c(pq,r)+c(p,qr)-c(p,q)=0,
	\]
	hence \(c\in Z^2(P,\Pi_A)\).
	
	If the representatives are changed to \(\gamma'_p=(b_p,I)\gamma_p\), \(b_p\in\Pi\), then \(a'_p=b_p+a_p\), and the corresponding cocycle is
	\[
	c'(p,q)
	=
	c(p,q)+b_p+A_pb_q-b_{pq}.
	\]
	Thus \(c'\) and \(c\) differ by a coboundary, so the class \([c]\in H^2(P,\Pi_A)\) is independent of the chosen representatives.
	
	The extension \eqref{eq:cryst_extension} splits if and only if \([c]=0\). Indeed, if \([c]=0\), the representatives may be changed so that \(c(p,q)=0\) for all \(p,q\). Then \(\gamma_p\gamma_q=\gamma_{pq}\), and the representatives form a subgroup of \(\Gamma\) isomorphic to \(P\). Conversely, if the extension splits, representatives can be chosen inside such a subgroup, and then \(c=0\).
	
	It remains to check the action on \(T_\Pi\). The formula
	\[
	\rho(p)[v]=[A_pv+a_p]
	\]
	is well defined with respect to \(v\): if \(v\) is replaced by \(v+\lambda\), \(\lambda\in\Pi\), then \(A_p\lambda\in\Pi\). It is also independent of the chosen representative \(\gamma_p\), since replacing \(a_p\) by \(a_p+b_p\), \(b_p\in\Pi\), does not change the class in \(\mathbb R^d/\Pi\). Finally,
	\[
	\begin{aligned}
		\rho(p)\rho(q)[v]
		&=[A_p(A_qv+a_q)+a_p]  \\
		&=[A_{pq}v+a_{pq}+c(p,q)] \\
		&=[A_{pq}v+a_{pq}]
		=\rho(pq)[v],
	\end{aligned}
	\]
	because \(c(p,q)\in\Pi\) vanishes in the quotient. Therefore \(\rho\) is a left affine action of \(P\) on \(T_\Pi\).
    If the extension splits, the chosen representatives form a finite group
    of Euclidean isometries. Averaging the orbit of any point gives a common
    fixed point. Taking this point as origin makes all their translation
    parts zero, proving the final assertion.
\end{proof}

\medskip

The preceding construction can be phrased cohomologically in a way that makes precise in what sense the affine torus action records the nonsymmorphic data. 
Let
\[
0\longrightarrow \Pi_A \longrightarrow \mathbb R^d_A \longrightarrow T_\Pi \longrightarrow 0
\]
be the exact sequence of \(P\)-modules, where \(P\) acts on \(\mathbb R^d\) and \(\Pi\) by the linear representation \(A\). Here \(T_\Pi\) is regarded as a \(P\)-module through the same linear action induced by \(A\); the affine translational parts then define a cocycle with values in this linear \(P\)-module.
The classes \(\bar a_p=[a_p]\in T_\Pi\) form a \(1\)-cocycle with values in \(T_\Pi\). 
Indeed, equation \eqref{eq:extension_cocycle} gives
\[
\bar a_{pq}=\bar a_p+A_p\bar a_q.
\]
The connecting homomorphism
\[
\delta:H^1(P,T_\Pi)\longrightarrow H^2(P,\Pi_A)
\]
associated with the above exact sequence sends the class of \(\bar a\) to the extension class \([c]\) represented by \eqref{eq:extension_cocycle}. 
Since \(P\) is finite and \(\mathbb R^d_A\) is a real \(P\)-module, averaging cochains gives
\[
H^n(P,\mathbb R^d_A)=0,
\qquad n>0.
\]
Consequently, in the relevant part of the long exact sequence, the connecting homomorphism identifies the affine translational class in \(H^1(P,T_\Pi)\) with the crystallographic extension class in \(H^2(P,\Pi_A)\). 
Thus the phrase ``the affine action records the extension class'' means that the translational part of the affine action determines, through this connecting map, the cohomology class distinguishing the nonsymmorphic extension from a split semidirect product.

\begin{remark}[Practical consequence for the translational order parameter]
	\label{rem:nonsymmorphic_affine_action}
	If only the linear action \([v]\mapsto[A_pv]\) is used below, then the theory is effectively restricted to the symmorphic case, or equivalently it forgets the extension class \([c]\). 
	For a general crystallographic group, the action of \(P\) on \(T_\Pi\) must be understood as the affine action \eqref{eq:affine_torus_action}. 
	Its differential is the linear action \(A_p\). 
	Thus the vertical tangent transition maps depend on \(A_p\), whereas the nonsymmorphic shifts \(a_p\) enter the gluing laws for the local representatives of the torus-valued field. 
	This distinction is built into the definition of \(Y_\Pi=Q\times_P T_\Pi\).
\end{remark}

\medskip

\noindent\textbf{Example: screw-axis gluing.}
Let the lattice contain a primitive vector \(e_z\) along an axis, and let \(R\) be a rotation of order \(n\) about this axis. 
A screw generator of type \(n_m\) may be represented by
\begin{equation}
	\gamma_s=\left(\frac{m}{n}e_z,R\right),
	\qquad
	R^n=I,
	\label{eq:screw_generator}
\end{equation}
with \(m e_z\in\Pi\). Then
\begin{equation}
	\gamma_s^{\,n}=(m e_z,I)\in\Pi,
	\label{eq:screw_power_translation}
\end{equation}
so the image \(s\in P\) has order \(n\), while its representative in \(\Gamma\) contains a fractional translation. The induced action on the translation torus is
\begin{equation}
	\rho(s)[v]=\left[Rv+\frac{m}{n}e_z\right],
	\qquad
	\rho(s)^n=\mathrm{id}_{T_\Pi}.
	\label{eq:screw_affine_torus_action}
\end{equation}
Thus, on an overlap whose transition function is \(s\), a local lift of the torus-valued displacement changes by the rotation \(R\), by the fractional shift \((m/n)e_z\), and by a lattice vector. 
Differentiating removes the constant fractional shift and leaves only the linear action of \(R\) on vertical tangent vectors. 
This removes the affine shift from tangent transition maps. For phase-independent quadratic coefficients, the fluctuation operator is independent of the reference section, as proved in Theorem~\ref{thm:global_hessian}; more general constitutive laws can retain indirect background dependence.

\medskip

\noindent\textbf{Scope of the main construction.}
The construction uses the crystallographic extension (including its affine action), the translation lattice \(\Pi\), and a fixed principal \(P\)-bundle \(Q\). 
Magnetic or time-reversing symmetries would require replacing \(P\) by the appropriate magnetic point group, and are not considered here.
\medskip

\section{The Order Parameter as a Geometric Object}
\label{sec:orderparameter}

For a principal $G$-bundle $\mathcal P\to X$ and a closed subgroup $H$,
sections of $\mathcal P/H\to X$ are equivalent to reductions of
$\mathcal P$ to $H$ \cite{HusemollerFibreBundles,KobayashiNomizu1}.
For a crystal the homogeneous order-parameter space is
$\mathrm E(d)/\Gamma$ \cite{KlemanMichel1978,Mermin1979}. Its orientational
and translational variables fit into the following fibration; no global
product is assumed.

\begin{proposition}[Torus fibration of the crystalline order-parameter space]
\label{prop:order_parameter_fibration}
The map
\begin{equation}
 \pi_{\mathrm{or}}:\mathrm E(d)/\Gamma\longrightarrow O(d)/P,
 \qquad [(b,B)]\longmapsto[B]
 \label{eq:order_parameter_fibration}
\end{equation}
is a smooth fiber bundle with fiber $T_\Pi$. With the affine action
$\rho$ of Proposition~\ref{prop:affine_point_action},
\begin{equation}
 \mathrm E(d)/\Gamma\cong O(d)\times_P T_\Pi,
 \qquad (B,\tau)\sim(BA_p,\rho(p^{-1})\tau).
 \label{eq:homogeneous_associated_torus_bundle}
\end{equation}
\end{proposition}
\begin{proof}
Right multiplication by an element of $\Gamma$ changes the orthogonal
part by an element of $P$, so \eqref{eq:order_parameter_fibration} is
well defined. Set
\[
 F(B,[v])=[(-Bv,B)].
\]
Changing $v$ by a lattice vector does not change this coset. For
$\gamma_p=(a_p,A_p)$, the identity
$a_{p^{-1}}\equiv-A_p^{-1}a_p\pmod\Pi$ gives
\[
 F(BA_p,\rho(p^{-1})[v])
 =[(-Bv+Ba_p,BA_p)]=[(-Bv,B)\gamma_p].
\]
Thus $F$ descends to \eqref{eq:homogeneous_associated_torus_bundle}.
It is onto since $(b,B)=(-Bv,B)$ for $v=-B^{-1}b$. The same calculation
shows that its fibers are exactly the stated equivalence classes.
Local quotient charts give a smooth inverse, and the map intertwines
the projections to $O(d)/P$.
\end{proof}

In the symmorphic case the fibration can be represented with linear
transition maps. In general its affine transitions retain the extension
class. Changing the origin conjugates $\rho$ by a torus translation and
changes the fibration only by a bundle isomorphism.

For the spatial interpretation fix a Riemannian manifold $(M,g)$.
Orientational order is a section of
\begin{equation}
 \mathcal Y_P=F_g(M)\times_{O(d)}(O(d)/P)\cong F_g(M)/P.
 \label{eq:YP_bundle}
\end{equation}
By the reduction theorem, such a section is equivalent to a principal
$P$-subbundle
\begin{equation}
 Q\subset F_g(M).
 \label{eq:Q_in_frame_bundle}
\end{equation}
Its existence is a genuine global condition, not a property of every
Riemannian manifold. When $P\subset SO(d)$ and an orientation is chosen,
one may instead use the oriented frame bundle. On
$X=\mathbb R_t\times M$ all spatial bundles are pulled back along
$\mathrm{pr}_M$.

\begin{remark}[Fixed orientational background]
In this paper $Q$ is prescribed and time independent. The varied field is
the translational phase, a section of $Y_\Pi=Q\times_P T_\Pi$.
The connection with the local acoustic Goldstone description is standard
\cite{LeutwylerPhonons,WatanabeMurayama2013,AndersenBraunerHofmannVuorinen,HayataHidaka2014}.
The inverse-Higgs mechanism and rotational redundancy
\cite{IvanovOgievetsky1975,LowManohar,BraunerInverseHiggs} motivate this
choice of sector but are not derived by freezing $Q$.
\end{remark}

\section{The Associated Bundle for a Discrete Structure Group}
\label{sec:associated_bundle}

Let $Q\to X$ be a fixed principal $P$-bundle and let
$\rho(p)[v]=[A_pv+a_p]$ be the action from
Proposition~\ref{prop:affine_point_action}. The abstract constructions in
this section do not require a frame reduction.

\begin{definition}[Configuration bundle]
The affine torus bundle is
\begin{equation}
 Y_\Pi=Q\times_P T_\Pi,
 \qquad (q,\tau)\sim(qp,\rho(p^{-1})\tau),
 \qquad \pi_Y([q,\tau])=\pi_Q(q).
 \label{eq:Ypi_associated}
\end{equation}
A configuration is a smooth section $\sigma:X\to Y_\Pi$.
\end{definition}

Choose local sections $q_\alpha$ of $Q$ on a good cover with
$q_\beta=q_\alpha g_{\alpha\beta}$. Fiber coordinates satisfy
$\tau_\beta=\rho(g_{\alpha\beta}^{-1})\tau_\alpha$.
For local real lifts of a configuration this means
\begin{equation}
 u_\beta=A_{g_{\alpha\beta}}^{-1}u_\alpha
          +a_{g_{\alpha\beta}^{-1}}+\lambda_{\alpha\beta},
 \qquad \lambda_{\alpha\beta}\in\Pi.
 \label{eq:lift_transition_compact}
\end{equation}
All transitions and lattice corrections are locally constant in the
base variables. Their differentials involve only the linear parts $A_p$.
In particular, with
\begin{equation}
 E_\Pi=Q\times_{P,A}\mathbb R^d,
 \qquad VY_\Pi\cong\pi_Y^*E_\Pi,
 \label{eq:vertical_vector_bundle}
\end{equation}
the tangent sequence is
\begin{equation}
 0\longrightarrow VY_\Pi\longrightarrow TY_\Pi
   \xrightarrow{T\pi_Y}\pi_Y^*TX\longrightarrow0.
 \label{eq:tangent_exact_sequence_compact}
\end{equation}
Here the vertical identification follows from
$TT_\Pi\cong T_\Pi\times\mathbb R^d$ and $d\rho(p)=A_p$.

\subsection{Horizontal Splitting and the Ehresmann Connection}
\label{subsec:ehresmann}

\begin{proposition}[Canonical flat connection]
\label{prop:ehresmann}
Relative to the fixed $P$-bundle $Q$, the distributions
$TU_\alpha\oplus0$ on $U_\alpha\times T_\Pi$ glue to a canonical flat
Ehresmann connection $HY_\Pi$ with
\begin{equation}
 TY_\Pi=HY_\Pi\oplus VY_\Pi.
 \label{eq:horizontal_splitting_compact}
\end{equation}
\end{proposition}
\begin{proof}
On each connected overlap the transition has the form
$(x,\tau)\mapsto(x,\phi_{\alpha\beta}(\tau))$ with
$\phi_{\alpha\beta}$ independent of $x$. Its differential sends
$(v,0)$ to $(v,0)$, so the horizontal subspaces agree.
They complement the vertical tangent spaces in every product chart.
Finally, basic horizontal vector fields satisfy
$[(X,0),(Y,0)]=([X,Y],0)$; hence the glued distribution is integrable.
\end{proof}

The word canonical refers to $Q$ and its flat structure. It does not
choose a connection independently of $Q$, nor does it identify this
connection with the Levi--Civita connection. The latter identification
requires the compatibility hypothesis of
Section~\ref{subsec:global_small_strain}.

\subsection{The Covariant Differential in the Jet Formulation}
\label{subsec:jet_covariant}

The first jet bundle $J^1Y_\Pi\to Y_\Pi$ is affine, modeled on
$\pi_Y^*T^*X\otimes VY_\Pi$. An Ehresmann connection is equivalently a
section $\Gamma_Y:Y_\Pi\to J^1Y_\Pi$
\cite{SardanashvilyBook,GiachettaBook,SaundersJetBundles}. Subtracting that
section defines
\begin{equation}
 D_{\Gamma_Y}:J^1Y_\Pi\longrightarrow\pi_Y^*T^*X\otimes VY_\Pi,
 \qquad
 D_{\Gamma_Y}=(\tau^i_\lambda-\Gamma^i_\lambda)\,dx^\lambda\otimes\partial_i.
 \label{eq:jet_covariant_differential_compact}
\end{equation}
For a section $\sigma$, this gives the vertical-valued form
$\nabla^{\Gamma_Y}\sigma=D_{\Gamma_Y}\circ J^1\sigma$ along $\sigma$.
The connection of Proposition~\ref{prop:ehresmann}, denoted by $\Gamma_0$,
has coefficients $\Gamma^i_{0\lambda}=0$ in admissible product charts.

The same locally constant linear transitions define the flat vector
connection
\begin{equation}
 D^E:\Gamma(X,E_\Pi)\longrightarrow\Omega^1(X;E_\Pi),
 \qquad (D^E\varphi)_\alpha=d\varphi_\alpha.
 \label{eq:flat_vector_connection}
\end{equation}
Using $\sigma^*VY_\Pi\cong E_\Pi$, set
\begin{equation}
 \beta_\sigma:=\nabla^{\Gamma_0}\sigma\in\Omega^1(X;E_\Pi).
 \label{eq:beta_definition}
\end{equation}
Thus $D^E$ is a linear connection on the vector bundle, whereas
$\beta_\sigma$ is the covariant differential of an affine torus section.

\begin{proposition}[Global translational order parameter and covariant differential]
\label{prop:global_phonon_covariant_differential}
Every existing configuration $\sigma\in\Gamma(X,Y_\Pi)$ has the globally
defined differential $\beta_\sigma$. In local real lifts it is $du_\alpha$;
these local forms transform only by the linear point-group action.
\end{proposition}
\begin{proof}
Differentiating \eqref{eq:lift_transition_compact} gives
$du_\beta=A_{g_{\alpha\beta}}^{-1}du_\alpha$, because both remaining
terms are locally constant. These are the transitions of $E_\Pi$.
\end{proof}

\subsection{Smooth sections and the pulled-back extension class}
\label{subsec:smooth_sections}

Let $\Pi_Q$ denote the local system of lattices $Q\times_{P,A}\Pi$
in the flat vector bundle $E_\Pi$ from \eqref{eq:vertical_vector_bundle}. The linear torus group bundle
\[
\mathcal T_Q=E_\Pi/\Pi_Q=Q\times_{P,A}T_\Pi
\]
acts freely and transitively on each fiber of $Y_\Pi$. Thus $Y_\Pi$ is a
$\mathcal T_Q$-torsor; it has no distinguished zero in general. This
linear torus bundle must be distinguished from the affine configuration
bundle. Let $f_Q:X\to BP$ classify $Q$, and let
$[c]\in H^2(P;\Pi_A)$ be the extension class of
\eqref{eq:cryst_extension}. Cohomology of $X$ below has local coefficients
in $\Pi_Q$. The following is the standard affine torus obstruction,
specialized to the crystallographic extension
\cite{BaragliaTorusMonodromy,BrownCohomology}.

\begin{theorem}[Smooth configurations and their homotopy classes]
\label{thm:smooth_sections}
There is a characteristic class
\begin{equation}
 o(Y_\Pi)=f_Q^*[c]\in H^2(X;\Pi_Q)
 \label{eq:smooth_section_obstruction}
\end{equation}
such that $Y_\Pi$ has a global smooth section if and only if $o(Y_\Pi)=0$.
If sections exist, their homotopy classes through sections form a torsor
for $H^1(X;\Pi_Q)$.
\end{theorem}

\begin{proof}
Choose a good cover and sections $q_\alpha$ of $Q$ with
$q_\beta=q_\alpha g_{\alpha\beta}$. Set
$A_{\alpha\beta}=A_{g_{\alpha\beta}}$ and
$a_{\alpha\beta}=a_{g_{\alpha\beta}}$. A configuration with local lifts
$u_\alpha$ satisfies
\begin{equation}
 u_\alpha=A_{\alpha\beta}u_\beta+a_{\alpha\beta}
                  +\ell_{\alpha\beta},\qquad \ell_{\alpha\beta}\in\Pi.
 \label{eq:section_lifts_obstruction}
\end{equation}
On a triple overlap define
\[
 c_{\alpha\beta\gamma}
 =a_{\alpha\beta}+A_{\alpha\beta}a_{\beta\gamma}
                         -a_{\alpha\gamma}\in\Pi.
\]
This is the pullback representative of $[c]$. Consistency of
\eqref{eq:section_lifts_obstruction} requires
\[
 c_{\alpha\beta\gamma}+\ell_{\alpha\beta}
 +A_{\alpha\beta}\ell_{\beta\gamma}-\ell_{\alpha\gamma}=0.
\]
Hence a section implies $o(Y_\Pi)=0$. Conversely, if this class vanishes,
choose $\ell_{\alpha\beta}$ satisfying the last identity. Then
$b_{\alpha\beta}=a_{\alpha\beta}+\ell_{\alpha\beta}$ is a twisted
$1$-cocycle in the sheaf of smooth sections of $E_\Pi$. This sheaf is fine,
so $b_{\alpha\beta}=u_\alpha-A_{\alpha\beta}u_\beta$ for smooth
$u_\alpha$. Their classes modulo $\Pi$ give a section.

To classify homotopy classes, choose one section $\sigma_*$. Every other
section is uniquely $\sigma_*+t$ for $t\in\Gamma(X,\mathcal T_Q)$.
The exact sequence of sheaves
\[
0\longrightarrow\Pi_Q\longrightarrow\mathscr E_\Pi
 \longrightarrow\mathscr T_Q\longrightarrow0
\]
and fineness of $\mathscr E_\Pi$ give the exact sequence
\begin{equation}
\Gamma(X,E_\Pi)\longrightarrow\Gamma(X,\mathcal T_Q)
 \longrightarrow H^1(X;\Pi_Q)\longrightarrow0.
\label{eq:homotopy_sections_exact}
\end{equation}
A torus section with a global lift $v\in\Gamma(X,E_\Pi)$ is homotopic
to zero by $[tv]$, $0\le t\le1$. Conversely, a homotopy from zero lifts
through the fiberwise covering $E_\Pi\to\mathcal T_Q$, giving a global
lift at its endpoint. Thus the quotient in
\eqref{eq:homotopy_sections_exact} is exactly the set of vertical homotopy
classes. The choice of $\sigma_*$ makes it a group; intrinsically it is a
torsor for that group.
\end{proof}

\subsection{Parallel reference sections and affine monodromy}
\label{subsec:monodromy_equilibrium}

Assume $X$ is connected, and choose $x_0\in X$ and $q_0\in Q_{x_0}$.
The monodromy representation, with a fixed convention for loop composition,
is
\begin{equation}
 h_Q:\pi_1(X,x_0)\longrightarrow P,
 \qquad H_Q=\operatorname{im}h_Q.
 \label{eq:principal_monodromy}
\end{equation}
Changing $q_0$ conjugates the representation and identifies the corresponding
fixed-point sets. In particular their nonemptiness is intrinsic.

\begin{proposition}[Parallel sections]
\label{prop:parallel_sections_monodromy}
Evaluation at the chosen fiber identifies the parallel sections with
\begin{equation}
 T_\Pi^{\rho(H_Q)}
 =\{\tau\in T_\Pi:\rho(h)\tau=\tau\text{ for all }h\in H_Q\}.
 \label{eq:affine_monodromy_fixed_set}
\end{equation}
\end{proposition}
\begin{proof}
On the universal cover the flat associated bundle is a product, whose
parallel sections are constant. A constant descends precisely when its
value is fixed by the monodromy. The fixed-point condition is unchanged
if all holonomy elements are replaced by their inverses.
\end{proof}

\begin{proposition}[Parallel-section obstruction]
\label{prop:cohomological_equilibrium_obstruction}
A parallel section exists if and only if
\begin{equation}
 [\bar a|_{H_Q}]=0\ \text{in }H^1(H_Q;T_\Pi),
 \quad\text{equivalently}\quad
 [c]|_{H_Q}=0\ \text{in }H^2(H_Q;\Pi_A).
 \label{eq:restricted_affine_class_zero}
\end{equation}
Here $T_\Pi$ is a module through the linear action $A$.
If $H_Q=P$, parallel sections exist exactly when the extension splits.
\end{proposition}
\begin{proof}
The fixed-point equation is $\bar a_h=\tau-A_h\tau$, which says that
$\bar a|_{H_Q}$ is a coboundary. Since $H_Q$ is finite, the exact sequence
$0\to\Pi_A\to\mathbb R_A^d\to T_\Pi\to0$ and
$H^n(H_Q;\mathbb R_A^d)=0$ for $n>0$ identify the two classes.
The last statement follows from Proposition~\ref{prop:affine_point_action}.
\end{proof}

For completeness, the implication from the parallel-section obstruction to
the smooth-section obstruction can also be seen directly at the level of
classifying maps. Let $Q_0$ be the connected component of $Q$ containing
$q_0$. Path lifting shows that $Q_0\to X$ is a principal $H_Q$-bundle, and
extension of structure group gives
$Q\simeq Q_0\times_{H_Q}P$. Consequently the classifying map
$f_Q:X\to BP$ factors up to homotopy through the map $BH_Q\to BP$ induced by
the inclusion $H_Q\hookrightarrow P$. Hence
\[
 f_Q^*[c]=\widetilde f_Q^{\,*}\bigl([c]|_{H_Q}\bigr),
\]
where $\widetilde f_Q:X\to BH_Q$ classifies $Q_0$. Therefore
$[c]|_{H_Q}=0$ implies $f_Q^*[c]=0$.

These are nevertheless different obstructions: smooth sections are
controlled by $f_Q^*[c]$ on $X$, whereas parallel sections are controlled by
restriction to $H_Q$. The converse implication need not hold.
For example every affine torus bundle of this type over $S^1$ has a smooth
section, since $H^2(S^1;\Pi_Q)=0$, while a fixed-point-free screw action
has no parallel section. This circle statement is an abstract bundle
example, or the restriction of a higher-dimensional bundle to a loop;
it is not a claim about a rank-three reduction of $F_g(S^1)$.
Section~\ref{sec:halfturn} realizes the distinction on an actual compatible
three-dimensional frame reduction.

A nontrivial principal $\mathbb Z_2$-bundle over $S^1$ with the sign action
also gives a valid abstract twisted torus bundle, but its associated
M\"obius line bundle cannot equal $TS^1$. Consequently it cannot serve as
a one-dimensional example with $Q\subset F_g(S^1)$.

Neither obstruction is a criterion for all stationary points of an energy.
Parallel sections have zero covariant differential and are stationary for
the quadratic derivative-only model below. Nonparallel stationary sections
can also exist. Even on a compatible background, zero strain means a
vanishing symmetric gradient, a weaker local condition than parallelness.

\subsection{Global fluctuations about arbitrary sections}
\label{subsec:toric_lift_linearization}

The vertical bundle has the canonical identification
\[
 VY_\Pi\simeq\pi_Y^*E_\Pi.
\]
Thus every existing smooth section $\sigma_0$ has
$\sigma_0^*VY_\Pi\simeq E_\Pi$, regardless of parallelness.
The infinitesimal fluctuation is globally
\begin{equation}
 \varphi\in\Gamma(X,E_\Pi).
 \label{eq:linear_fluctuation_section}
\end{equation}
The torsor action defines $\sigma_0+\varphi$ by first reducing $\varphi$
modulo $\Pi_Q$. In local lifts it is represented by $u_{0,\alpha}+\varphi_\alpha$.
Both configurations have the same affine gluing; their difference transforms
linearly. Consequently
\begin{equation}
 \beta_{\sigma_0+\varphi}=\beta_{\sigma_0}+D^E\varphi.
 \label{eq:affine_derivative_identity}
\end{equation}
This is exact, not merely infinitesimal. Fiberwise small $\varphi$
gives local torus coordinates near $\sigma_0$; arbitrary global
$\varphi$ remains meaningful via the torus quotient. By
Theorem~\ref{thm:smooth_sections}, every section in the same vertical
homotopy class as $\sigma_0$ is representable this way. The lift is not
unique: differences between lifts are global parallel lattice sections.

The absence of a parallel section therefore does not prevent global
linearization about a smooth stationary section. It only prevents dropping
$\beta_{\sigma_0}$ from \eqref{eq:affine_derivative_identity}.
Moreover, if $Q\subset F_g(M)$, $E_\Pi\simeq\mathrm{pr}_M^*TM$.
A fluctuation is then a global spatial vector field. On a Euclidean domain
it has global Cartesian components even when crystallographic components
are twisted. The connection and operator retain the twisting information
under a change of frame.

\subsection{Soldering, compatibility and linearized objectivity}
\label{subsec:global_small_strain}

For a frame reduction the soldering map is
\begin{equation}
 s:E_\Pi\xrightarrow{\sim}\mathrm{pr}_M^*TM,
 \qquad [q,v]\longmapsto qv.
 \label{eq:soldering_Epi_TM}
\end{equation}
The induced metric connection on $TM$ is
$D^Q=s\circ D^E\circ s^{-1}$ (spatially). In a local $Q$-frame
$q=(q_i)$ one has $D^Qq_i=0$. Therefore it is flat, but its torsion is
\begin{equation}
 T^{D^Q}(q_i,q_j)=-[q_i,q_j].
 \label{eq:induced_torsion}
\end{equation}
A smooth discrete reduction does not make the frame holonomic.
In particular flatness of $D^Q$ alone is insufficient for the usual
compatible elastic interpretation.

For a fluctuation with components $\varphi$ and physical vector field
$U=q\varphi$, one has
\begin{equation}
 \nabla^gU=q\,d\varphi+(\nabla^gq)\varphi.
 \label{eq:connection_strain_comparison}
\end{equation}
Thus $\operatorname{sym}((q\,d\varphi)^\flat)$ need not be the ordinary
small strain $\frac12\mathcal L_Ug$. For example, on the Euclidean plane
choose $q(x,y)=R(\kappa x)$ and $U=e_x$, with $R$ the planar rotation.
Then $\varphi=q^{-1}e_x$ and $q\,d\varphi=-\kappa e_y\,dx$, so the
symmetric gradient defined using $D^Q$ is nonzero even though $U$ is a
rigid translation. This excludes the identification for arbitrary $Q$.

\begin{definition}[Compatible crystallographic background]
\label{def:compatible_background}
A fixed reduction $Q\subset F_g(M)$ is called compatible here if it is
preserved by Levi--Civita parallel transport. The time-independent
pullback is used on $\mathbb R_t\times M$.
\end{definition}
Under this hypothesis the connection form of $\nabla^g$ in a $Q$-frame
takes values in the zero Lie algebra of $P$. Hence
\begin{equation}
 D^Q=\nabla^g,\qquad R^g=0,\qquad T^{D^Q}=0.
 \label{eq:compatible_connections}
\end{equation}
Local $Q$-frames are parallel and admit Cartesian coordinates after shrinking
the chart. Finite global holonomy can still be nontrivial on a flat quotient.
In contrast, on a connected domain in $\mathbb R^d$ with the Euclidean
metric, Cartesian parallel transport has trivial holonomy even when the
domain is not simply connected. Any compatible $Q$ then has trivial
monodromy. Nonsymmorphic microscopic symmetry alone therefore does not
produce twisted acoustic boundary conditions for an ordinary Euclidean
body. The example in Section~\ref{sec:halfturn} uses actual screw-periodic
spatial identifications.

On a compatible background, for
$\beta_\sigma=\nabla^{\Gamma_0}\sigma$ and its spatial part $F_\sigma$,
define
\begin{equation}
 \varepsilon(\sigma)
 =\operatorname{sym}\bigl((s\circ F_\sigma)^\flat\bigr)
 \in\Gamma\bigl(X,\mathrm{pr}_M^*\operatorname{Sym}^2T^*M\bigr).
 \label{eq:global_small_strain}
\end{equation}
Its local expression in parallel Cartesian crystallographic frames is
$\varepsilon_{ij}=\frac12(\partial_i u_j+\partial_j u_i)$, and
\begin{equation}
 \varepsilon(\sigma_0+\varphi)
 =\varepsilon(\sigma_0)+\operatorname{sym}((\nabla^gU)^\flat).
 \label{eq:strain_variation_identity}
\end{equation}
This identity concerns local displacement lifts and global variations;
it does not assert a global real lift of $\sigma$ itself.

The leading spatial energy is assumed to depend only on this symmetric
strain. This is \emph{linearized objectivity}, eliminating the contribution
of infinitesimal rigid rotations. It is not a claim of full finite-strain
frame indifference. Material $P$-symmetry is a separate constitutive
assumption: one chooses a $P$-invariant quadratic form on
$\operatorname{Sym}^2(\mathbb R^d)^*$ and transports it by $Q$.
Tensorial covariance of an arbitrary globally defined coefficient field
does not by itself impose this material symmetry.

\section{First-Order Lagrangian Formulation}
\label{sec:first_order_lagrangian}

\subsection{Jets, densities and constitutive restrictions}

Let $\operatorname{Dens}(X)=|\Lambda^{\dim X}T^*X|$. A first-order
Lagrangian is a morphism $L:J^1Y_\Pi\to\operatorname{Dens}(X)$ over $X$.
For a compact domain $D$ its action is
\[
 S_D[\sigma]=\int_D L\circ J^1\sigma.
\]
We use compactly supported variations in the interior, or boundary
conditions for which the indicated boundary term vanishes. The jet and
Euler--Lagrange formalism is standard
\cite{SaundersJetBundles,SardanashvilyBook,GiachettaBook}.

A fixed connection identifies the affine jet bundle with its vector model:
\begin{equation}
 D_\Gamma:J^1Y_\Pi\xrightarrow{\sim}
       \pi_Y^*T^*X\otimes VY_\Pi,
 \qquad (x,\tau,\tau_\lambda)\mapsto
 (x,\tau,\tau_\lambda-\Gamma_\lambda).
 \label{eq:lagrangian_factorization}
\end{equation}
Consequently \emph{every} first-order Lagrangian can be written
$L=\widetilde L\circ D_\Gamma$. This is a change of variables, not a
restriction to derivative-only physics: the target still remembers $\tau$.
For the canonical connection, the actual phase-independent restriction is
factorization through
\[
 J^1Y_\Pi\longrightarrow T^*X\otimes E_\Pi,
 \qquad j^1_x\sigma\longmapsto\beta_\sigma(x),
\]
which forgets the value of the affine torus coordinate.
Constant torus shifts express this restriction locally; in a twisted bundle
not every such constant shift is a single global symmetry.

For comparison, choose a positive density $\nu=w(x)|dx|$ and write a general
quadratic local Lagrangian as
\begin{equation}
 L=\frac12G^{\lambda\mu}_{ij}(x,\tau)
       \xi^i_\lambda\xi^j_\mu\,\nu,
 \qquad \xi^i_\lambda=\tau^i_\lambda-\Gamma^i_\lambda(x,\tau),
 \label{eq:quadratic_lagrangian_local}
\end{equation}
with $G^{\lambda\mu}_{ij}=G^{\mu\lambda}_{ji}$. These coefficients must be
periodic in lifted torus coordinates and obey the affine-equivariant
tensor transformation law. Keeping the density coefficient explicit,
the equations are
\begin{equation}
\begin{split}
0={}&\frac12\partial_{\tau^k}G^{\lambda\mu}_{ij}
                         \xi^i_\lambda\xi^j_\mu
 -G^{\lambda\mu}_{ij}(\partial_{\tau^k}\Gamma^i_\lambda)\xi^j_\mu\\
 &-w^{-1}d_\kappa\bigl(wG^{\kappa\mu}_{kj}\xi^j_\mu\bigr).
\end{split}
\label{eq:quadratic_EL_explicit}
\end{equation}
Here $d_\kappa$ is the total derivative with respect to $x^\kappa$.
For $\Gamma_0$ and phase-independent coefficients this reduces to
\begin{equation}
 w^{-1}\partial_\lambda
       \bigl(wG^{\lambda\mu}_{ij}(x)\partial_\mu u^j\bigr)=0.
 \label{eq:quadratic_EL_tau_independent}
\end{equation}
The principal symbol naturally maps $E_{\Pi,x}$ to its dual (with a density
factor if the density has not been trivialized); its components, up to
overall sign convention, are
$G^{\lambda\mu}_{ij}\zeta_\lambda\zeta_\mu$.
An endomorphism notation requires a fiber metric or another identification
with the dual.

\subsection{Exact expansion and global second variation}
\label{subsec:quadratic_phonon_lagrangians}

Let $\mathbf B$ be a smooth symmetric bilinear form on $T^*X\otimes E_\Pi$
over $X$, with no dependence on the affine torus phase. Consider
\begin{equation}
 S_D[\sigma]=\frac12\int_D
      \mathbf B(\beta_\sigma,\beta_\sigma)\,\nu.
 \label{eq:phase_independent_action}
\end{equation}
The form may be indefinite for a spacetime action. No positivity is needed
for the following identities.

\begin{theorem}[Exact global fluctuation identity and background-independent Hessian]
\label{thm:global_hessian}
For any smooth section $\sigma_0$ and any smooth variation
$\varphi\in\Gamma(X,E_\Pi)$, the torsor sum satisfies
\begin{equation}
\begin{split}
 S_D[\sigma_0+\varphi]={}&S_D[\sigma_0]
 +\int_D\mathbf B(\beta_{\sigma_0},D^E\varphi)\,\nu\\
 &+\frac12\int_D\mathbf B(D^E\varphi,D^E\varphi)\,\nu.
\end{split}
\label{eq:exact_quadratic_expansion}
\end{equation}
In particular, in the affine coordinates defined by the torsor action,
\begin{equation}
 \delta^2 S_{D,\sigma_0}(\varphi,\psi)
 =\int_D\mathbf B(D^E\varphi,D^E\psi)\,\nu.
 \label{eq:global_second_variation}
\end{equation}
This expression does not depend on $\sigma_0$. If $\sigma_0$ is stationary
for the admissible variations, the linear term in
\eqref{eq:exact_quadratic_expansion} vanishes. Parallelness is not required.
\end{theorem}
\begin{proof}
Equation~\eqref{eq:affine_derivative_identity} is obtained by differentiating
$u_{0,\alpha}+\varphi_\alpha$ in local lifts and glues by the linear action.
Substitution into \eqref{eq:phase_independent_action} and bilinearity give
\eqref{eq:exact_quadratic_expansion}. Substituting
$s\varphi+t\psi$ and taking the mixed derivative gives
\eqref{eq:global_second_variation}. Stationarity is precisely the vanishing
of the displayed linear functional on the allowed variations.
\end{proof}

The intrinsic Hessian at a stationary point is therefore represented by
\eqref{eq:global_second_variation}. Choose auxiliary metrics to express the
formal adjoint, and represent $\mathbf B$ by the corresponding coefficient
map $\mathbf B^\flat$. Then the weak Euler--Lagrange equation and Jacobi
operator are, respectively,
\begin{equation}
 (D^E)^*\mathbf B^\flat\beta_{\sigma_0}=0,
 \qquad J=(D^E)^*\mathbf B^\flat D^E.
 \label{eq:global_jacobi_operator}
\end{equation}
Equivalently these expressions may be defined directly by duality with the
chosen density, without choosing identifications with dual bundles.
Boundary conditions are part of an operator realization, not determined by
the principal symbol alone.

Within a homotopy class, every other configuration is $\sigma_0+\varphi$.
For a nonnegative static quadratic energy on a closed base, a stationary
$\sigma_0$ is thus a minimizer in that class; with boundary, this conclusion
only compares configurations whose lifted differences are admissible variations. Comparison between different classes requires
additional analysis. The theorem does not imply a background-independent
Hessian for phase-dependent coefficients or for nonlinear strain energies.
In particular, an affine shift may affect a general fluctuation operator
indirectly by forcing a different background.

\subsection{Compatible linear elasticity and local acoustic equations}
\label{subsec:linear_dispersion}

Assume now Definition~\ref{def:compatible_background}, with
$X=\mathbb R_t\times M$. Via the soldering isomorphism write a variation as
$U\in\Gamma(X,\mathrm{pr}_M^*TM)$ and set
\[
 e(U)=\operatorname{sym}((\nabla^g U)^\flat),\qquad
 v_\sigma=s\beta_\sigma(\partial_t).
\]
Choose a positive mass tensor $\boldsymbol\rho$ and a quadratic elastic
form $C$ on symmetric spatial covariant tensors. Their material $P$-invariance is an explicit constitutive
assumption. On a finite time interval $I=[t_-,t_+]$, the action is
\begin{equation}
 S[\sigma]=\frac12\int_{I\times M}
 \bigl(\boldsymbol\rho(v_\sigma,v_\sigma)
       -C(\varepsilon(\sigma),\varepsilon(\sigma))\bigr)\,dt\,dV_g.
 \label{eq:objective_action}
\end{equation}
We assume time-reversal symmetry of this leading model and no mixed
velocity--gradient term. This is a choice of representative Lagrangian;
terms which are total divergences do not affect the bulk equations.
Take variations vanishing at
the temporal endpoints (or compactly supported in its interior), and the
spatial boundary conditions specified above. The action is understood on
$I\times M$; this avoids an infinite time integral for a static strained
background. For any stationary background define the \emph{quadratic
fluctuation action} by
\[
 S^{(2)}[U]:=\tfrac12\,\delta^2S_{\sigma_0}(U,U).
\]
It is given by
\begin{equation}
 S^{(2)}[U]=\frac12\int_{I\times M}
 \bigl(\boldsymbol\rho(\partial_tU,\partial_tU)-C(e(U),e(U))\bigr)
 \,dt\,dV_g.
 \label{eq:elastic_second_variation}
\end{equation}
The background can be strained and nonparallel.

In parallel Cartesian coordinates, and for constant coefficients, write
\[
 \mathcal L^{(2)}
 =\frac12\rho_{ij}\dot U^i\dot U^j
  -\frac12 C^{ab}_{ij}\partial_aU^i\partial_bU^j.
\]
The elastic tensor here is the gradient representation of the quadratic
form on symmetric strains, so it has the usual minor and major symmetries
after identifying spatial and displacement indices. The equations are
\begin{equation}
 \rho_{ij}\partial_t^2U^j-C^{ab}_{ij}\partial_a\partial_bU^j=0.
 \label{eq:linearized_phonon_equation}
\end{equation}
For variable coefficients one must keep the divergence form; freezing
coefficients only describes the local principal symbol, not a global
plane-wave spectrum.

For a homogeneous medium the ansatz $U=\alpha e^{i(k\cdot x-\omega t)}$
gives
\begin{equation}
 \mathcal C(k)\alpha=\omega^2\rho\alpha,
 \qquad \mathcal C_{ij}(k)=C^{ab}_{ij}k_ak_b.
 \label{eq:christoffel_equation}
\end{equation}
If $\rho>0$ and $C^{ab}_{ij}\eta^i\eta^j\xi_a\xi_b>0$ for nonzero
$\eta,\xi$, then $\rho^{-1/2}\mathcal C(k)\rho^{-1/2}$ is symmetric
positive definite for $k\ne0$. Its eigenvalues are $\omega^2$, and
\begin{equation}
 \omega_\alpha(k)=c_\alpha(\hat k)|k|.
 \label{eq:acoustic_dispersion}
\end{equation}
This is a local statement. A global spectrum additionally requires a
stationary background, coefficient compatibility, an operator domain and
the complete spatial identifications.

In parallel crystallographic frames a prescribed $P$-invariant material law
satisfies, for $R\in P$,
\[
 \rho_{ij}=R_i{}^kR_j{}^l\rho_{kl},\qquad
 C^{ab}_{ij}=R^a{}_cR^b{}_dR_i{}^kR_j{}^lC^{cd}_{kl}.
\]
Hence $\mathcal C(Rk)=R\mathcal C(k)R^{-1}$. This material symmetry is
stronger than the tensorial transformation rule for an arbitrary coefficient
field on a bundle.

For the isotropic small-strain law set
\begin{equation}
 W(e)=\frac{\lambda_{\rm el}}2(\operatorname{tr}e)^2+\mu|e|^2,
 \qquad \rho_{ij}=\rho_m\delta_{ij}.
 \label{eq:isotropic_energy}
\end{equation}
The standard Christoffel matrix and speeds are
\begin{equation}
 \mathcal C(k)=\mu|k|^2I+(\lambda_{\rm el}+\mu)k\otimes k,
 \qquad c_T^2=\frac\mu{\rho_m},\quad
 c_L^2=\frac{\lambda_{\rm el}+2\mu}{\rho_m}.
 \label{eq:isotropic_speeds}
\end{equation}
Positive energy on all symmetric strains requires
$\mu>0$ and $d\lambda_{\rm el}+2\mu>0$, whereas strong ellipticity only
requires $\mu>0$ and $\lambda_{\rm el}+2\mu>0$ (for $d\ge2$).
These conditions should not be conflated. The compatible half-turn model
below uses positive strain energy in dimension three.

\section{A Compatible Half-Turn Model: Obstruction, Equilibrium and Spectrum}
\label{sec:halfturn}

\subsection{The spatial manifold and the two crystallographic scales}
\label{subsec:nonsymmorphic_loop_example}

Fix $L_x,L_y,L>0$ and $R=\operatorname{diag}(-1,-1,1)$. Let $\mathcal D$
be generated by the Euclidean isometries
\begin{equation}
\begin{aligned}
 T_x(x,y,z)&=(x+L_x,y,z),\\
 T_y(x,y,z)&=(x,y+L_y,z),\\
 S(x,y,z)&=(-x,-y,z+L).
\end{aligned}
\label{eq:halfturn_generators}
\end{equation}
They obey $ST_xS^{-1}=T_x^{-1}$, $ST_yS^{-1}=T_y^{-1}$ and
$S^2(x,y,z)=(x,y,z+2L)$. Set
\begin{equation}
 M=\mathbb R^3/\mathcal D.
 \label{eq:halfturn_manifold}
\end{equation}
Every element with rotational part $R$ has a nonzero odd multiple of $L$
as its axial translation, and every nonidentity element with rotational
part $I$ is a nonzero translation. Thus the action is free. Its translation
subgroup has finite index and full rank, so it is properly discontinuous
and cocompact. The quotient is a compact oriented flat manifold without
boundary, with holonomy $P=\{I,R\}$ and a double torus cover
\begin{equation}
 \widehat M=\mathbb R^3/
  (L_x\mathbb Ze_x\oplus L_y\mathbb Ze_y\oplus2L\mathbb Ze_z),
 \qquad \operatorname{Vol}(M)=L_xL_yL.
 \label{eq:halfturn_cover}
\end{equation}
This is the standard half-turn flat manifold
\cite{Charlap1986,EskiltHalfTurn}. Parallel Cartesian frames generate a
reduction $Q\subset F_g(M)$ preserved by $\nabla^g$. Its monodromy sends
$T_x,T_y$ to $I$ and $S$ to $R$, and
\[
 E_\Pi\simeq TM,\qquad D^E\simeq\nabla^g.
\]
In particular the metric compatibility condition of
Definition~\ref{def:compatible_background} holds.

The macroscopic group $\mathcal D$ specifies the topology and metric
identifications of physical space, and therefore the domain on which the
global elastic operator acts. It is conceptually distinct from the
microscopic crystallographic extension $\Gamma$, which specifies how the
translation phase is glued in the internal torus fiber. Keeping these two
inputs independent is deliberate: it separates effects due to the global
spatial quotient from effects due to nonsymmorphic affine translation data.
The resulting construction is a finite-volume global model used to isolate
these mechanisms; it is not intended as a microscopic model of a particular
material or as a claim that nonsymmorphicity by itself forces strain in an
ordinary Euclidean specimen.

Independently, choose the microscopic lattice and residual group
\begin{equation}
 \Pi=\ell_x\mathbb Ze_x\oplus\ell_y\mathbb Ze_y\oplus\ell_z\mathbb Ze_z,
 \qquad
 \Gamma=\left\langle\Pi,
        \gamma_s=\left(\frac{\ell_z}{2}e_z,R\right)\right\rangle,
 \label{eq:microscopic_screw_group}
\end{equation}
with $\ell_x,\ell_y,\ell_z>0$. The translation subgroup of $\Gamma$ is
exactly $\Pi$, its point group is $C_2$, and $\gamma_s^2=(\ell_ze_z,I)$.
It determines the affine action
\begin{equation}
 \rho(s)[u]=\left[Ru+\frac{\ell_z}{2}e_z\right].
 \label{eq:halfturn_affine_action}
\end{equation}
No equality of microscopic and macroscopic lengths is assumed. This is a
continuum phase-field model with screw-periodic spatial identifications;
it is not the boundary-value problem of a finite body embedded in
$\mathbb R^3$ with a free surface.

\subsection{A nonparallel configuration and the two obstruction classes}

The fixed-point equation for \eqref{eq:halfturn_affine_action}, projected
onto the axis, would require $0=\ell_z/2\pmod{\ell_z}$. Thus no parallel
section exists. Algebraically,
\begin{equation}
 H^2(C_2;\Pi_A)\simeq\Pi^{C_2}/(I+R)\Pi
 =\ell_z\mathbb Ze_z/(2\ell_z\mathbb Ze_z)\simeq\mathbb Z_2,
 \label{eq:halfturn_group_cohomology}
\end{equation}
and $c(s,s)=\ell_ze_z$ represents the nonzero element.

Nevertheless, put
\begin{equation}
 \alpha=\frac{\ell_z}{2L},\qquad
 u_0(x,y,z)=-\alpha z\,e_z.
 \label{eq:halfturn_reference}
\end{equation}
It is invariant under $T_x,T_y$ and obeys
\[
 u_0(Sx)=Ru_0(x)-\frac{\ell_z}{2}e_z
       \equiv Ru_0(x)+\frac{\ell_z}{2}e_z\pmod\Pi.
\]
Also $u_0(x,y,z+2L)=u_0(x,y,z)-\ell_ze_z$. These identities prove that
$[u_0]$ descends to a smooth section $\sigma_0$ of $Y_\Pi$. The sign of the
chosen lift differs from the positive screw translation by a lattice
vector and is consistent with the affine action.
Consequently
\begin{equation}
 \beta_{\sigma_0}=-\alpha\,dz\otimes e_z,\qquad
 f_Q^*[c]=0,\qquad [c]|_{H_Q}\ne0.
 \label{eq:halfturn_two_obstructions}
\end{equation}
The pulled-back class vanishes by Theorem~\ref{thm:smooth_sections}, since
we have explicitly constructed a section. Here $H_Q=C_2$, so the
restricted group class remains nonzero. This is a compatible spatial
example separating smooth and parallel configurations.

\subsection{Stationarity and minimization in a fixed sector}

Use the energy \eqref{eq:isotropic_energy} with constant coefficients and
assume
\begin{equation}
 \mu>0,\qquad 3\lambda_{\rm el}+2\mu>0,\qquad\rho_m>0.
 \label{eq:halfturn_positivity}
\end{equation}
The phase-field energy on $M$ is
\[
 \mathcal E[\sigma]=\int_M
 \left[\frac{\lambda_{\rm el}}2
      (\operatorname{tr}\varepsilon(\sigma))^2
      +\mu|\varepsilon(\sigma)|^2\right]dV_g.
\]

\begin{theorem}[A stable strained reference configuration]
\label{thm:halfturn_equilibrium}
The section $\sigma_0$ in \eqref{eq:halfturn_reference} is stationary and
minimizes $\mathcal E$ in its vertical homotopy class. Its energy is
\begin{equation}
 \mathcal E[\sigma_0]
 =\frac{L_xL_y(\lambda_{\rm el}+2\mu)\ell_z^2}{8L}.
 \label{eq:halfturn_background_energy}
\end{equation}
The minimizers in this class are precisely $\sigma_0+c e_z$, with the
constant shift interpreted modulo the axial lattice period.
\end{theorem}
\begin{proof}
The strain and stress of $\sigma_0$ are the globally parallel tensors
\[
 \varepsilon_0=-\alpha\,dz\otimes dz,\qquad
 \Sigma_0=-\lambda_{\rm el}\alpha g-2\mu\alpha\,dz\otimes dz.
\]
Thus $\operatorname{div}\Sigma_0=0$. Integration by parts on the closed
manifold gives $\delta\mathcal E_{\sigma_0}(U)=0$ for every smooth global
vector field $U$. Every section in its homotopy class is $\sigma_0+U$
by Theorem~\ref{thm:smooth_sections}. The exact expansion gives
\begin{equation}
 \mathcal E[\sigma_0+U]-\mathcal E[\sigma_0]
 =\frac12\int_M
  \bigl[\lambda_{\rm el}(\operatorname{div}U)^2+2\mu|e(U)|^2\bigr]dV_g.
 \label{eq:halfturn_minimization_identity}
\end{equation}
For any symmetric tensor $e$ in dimension three,
\[
 2\mu|e|^2+\lambda_{\rm el}(\operatorname{tr}e)^2
 =2\mu|\operatorname{dev}e|^2
  +(\lambda_{\rm el}+2\mu/3)(\operatorname{tr}e)^2.
\]
This is positive definite under \eqref{eq:halfturn_positivity}. Equality
in \eqref{eq:halfturn_minimization_identity} therefore requires $e(U)=0$.
On a closed flat manifold integration by parts gives
\[
 2\int_M|e(U)|^2\,dV_g
 =\int_M|\nabla U|^2\,dV_g+\int_M(\operatorname{div}U)^2\,dV_g.
\]
Hence such a $U$ is parallel. Its constant lift must be fixed by $R$,
so $U=ce_z$. Finally, the constant energy density is
$(\lambda_{\rm el}+2\mu)\alpha^2/2$; multiplying by
\eqref{eq:halfturn_cover} yields \eqref{eq:halfturn_background_energy}.
\end{proof}

This theorem compares configurations within the indicated sector.
Section~\ref{subsec:all_sectors} treats the other sectors and the global
energy bound. The strained state is within the
small-strain regime when $\ell_z/(2L)\ll1$. It is a static solution of
the time-dependent model \eqref{eq:objective_action}, even though a
parallel reference is topologically forbidden.

\subsection{All homotopy sectors and their minimum energies}
\label{subsec:all_sectors}

Fix $\sigma_0$ as the origin of the torsor of homotopy classes. The
fundamental group of $M$ has the presentation
\begin{equation}
 \mathcal D=\langle a,b,s\mid [a,b]=1,\;
 sas^{-1}=a^{-1},\;sbs^{-1}=b^{-1}\rangle,
 \label{eq:halfturn_group_presentation}
\end{equation}
where $a=T_x$, $b=T_y$ and $s=S$; in particular $s$ has infinite order.
The action on $\Pi$ is $A_a=A_b=I$, $A_s=R$.

\begin{proposition}[Explicit homotopy sectors]
\label{prop:halfturn_sectors}
For the rectangular microscopic lattice \eqref{eq:microscopic_screw_group},
\begin{equation}
 H^1(M;\Pi_Q)\cong\mathbb Z^5\oplus(\mathbb Z/2\mathbb Z)^2.
 \label{eq:halfturn_h1}
\end{equation}
With the origin $\sigma_0$, a class is specified by
$N=(n_{ij})\in\operatorname{Mat}_{2\times2}(\mathbb Z)$,
$m\in\mathbb Z$ and $\epsilon_x,\epsilon_y\in\{0,1\}$.
For any $c\in\mathbb R$, a real lift of a representative is
\begin{equation}
 u_{N,m,\epsilon,c}(x,y,z)=
 \begin{pmatrix}
 B_\perp\binom{x}{y}+t_\epsilon\\[1mm]
 \gamma_m z+c
 \end{pmatrix},
 \qquad
 t_\epsilon=\binom{\epsilon_x\ell_x/2}{\epsilon_y\ell_y/2},
 \label{eq:all_sector_representatives}
\end{equation}
where the upper block is a two-dimensional vector and
\begin{equation}
 B_\perp=
 \begin{pmatrix}
 \ell_x n_{11}/L_x&\ell_x n_{12}/L_y\\
 \ell_y n_{21}/L_x&\ell_y n_{22}/L_y
 \end{pmatrix},
 \qquad
 \gamma_m=\left(m-\tfrac12\right)\frac{\ell_z}{L}.
 \label{eq:sector_parameters}
\end{equation}
Write $\sigma_{N,m,\epsilon,c}$ for the section defined by the class
of this lift modulo $\Pi$. All smooth sections in the same homotopy class
are $\sigma_{N,m,\epsilon,0}+U$, with $U\in\Gamma(M,TM)$ interpreted by
the torsor action. The section depends on $c$ only modulo $\ell_z$, and
$c$ does not change the homotopy class.
\end{proposition}
\begin{proof}
The universal cover is contractible, so local-coefficient cohomology
agrees with group cohomology of $\mathcal D$. A crossed homomorphism
$f:\mathcal D\to\Pi$, satisfying $f(gh)=f(g)+A_gf(h)$, is determined by
$u=f(a)$, $v=f(b)$ and $w=f(s)$. The relations in
\eqref{eq:halfturn_group_presentation} give
\[
 (I+R)u=(I+R)v=0,
\]
and no constraint on $w$. Thus $u,v$ are arbitrary transverse lattice
vectors and $w\in\Pi$. A coboundary is zero on $a,b$ and changes $w$
by an element of $(R-I)\Pi$. Consequently the four transverse entries
of $u,v$ are free integers, while
\[
 \Pi/(R-I)\Pi\cong
 (\mathbb Z/2\mathbb Z)^2\oplus\mathbb Z,
\]
which proves \eqref{eq:halfturn_h1}.

For completeness, the correspondence with configurations can be checked
on the cover. The difference $[u-u_0]$ of any configuration and the
reference is a linear torus section; choose a real lift $v$ on
$\mathbb R^3$. The lattice differences
$v(gx)-A_gv(x)=f(g)$ are constant and form a crossed homomorphism.
Changing the real lift changes $f$ by a coboundary. Conversely, equal
cohomology classes allow a lattice change of lift making the two cocycles
equal; the difference is then an equivariant real vector field and
provides a homotopy through sections. This agrees with the classification
of Theorem~\ref{thm:smooth_sections}.
For \eqref{eq:all_sector_representatives} the lifted difference has
\[
\begin{aligned}
 f(a)&=(\ell_x n_{11},\ell_y n_{21},0),\\
 f(b)&=(\ell_x n_{12},\ell_y n_{22},0),\\
 f(s)&=(\epsilon_x\ell_x,\epsilon_y\ell_y,m\ell_z).
\end{aligned}
\]
These values realize every class. Directly, translation by $T_x$ or
$T_y$ changes $u$ by a lattice vector, and
\[
 u(Sx)-Ru(x)-\tfrac{\ell_z}{2}e_z
 =(\epsilon_x\ell_x,\epsilon_y\ell_y,(m-1)\ell_z)\in\Pi.
\]
Thus the proposed representatives have the required affine gluing.
The final assertion follows from the lifting statement in
Theorem~\ref{thm:smooth_sections}.
\end{proof}

\begin{theorem}[Equilibria and minima in every sector]
\label{thm:all_sector_minima}
Every representative \eqref{eq:all_sector_representatives} is stationary
and minimizes the quadratic energy in its homotopy class. Set
\begin{equation}
 a_N=\frac{\ell_xn_{11}}{L_x},\qquad
 b_N=\frac{\ell_yn_{22}}{L_y},\qquad
 h_N=\frac{\ell_xn_{12}}{L_y}+\frac{\ell_yn_{21}}{L_x}.
 \label{eq:sector_strain_parameters}
\end{equation}
The minimum energy in that class is
\begin{equation}
 E_{N,m}=L_xL_yL\left[
 \frac{\lambda_{\rm el}}2(a_N+b_N+\gamma_m)^2
 +\mu\left(a_N^2+b_N^2+\gamma_m^2+\frac{h_N^2}{2}\right)
 \right].
 \label{eq:all_sector_energy}
\end{equation}
All stationary sections in the class, and hence all its minimizers,
are precisely $\sigma_{N,m,\epsilon,c}$,
$c\in\mathbb R/\ell_z\mathbb Z$.
\end{theorem}
\begin{proof}
The lifted strain is the constant tensor with matrix
\[
 \varepsilon_{N,m}=
 \begin{pmatrix}
 a_N&h_N/2&0\\ h_N/2&b_N&0\\0&0&\gamma_m
 \end{pmatrix}.
\]
It is $R$-invariant and therefore globally parallel. Its isotropic
stress is also parallel, so integration by parts gives stationarity.
Substitution in the energy gives \eqref{eq:all_sector_energy}.
For any section in the same class, the exact expansion yields
\[
 \mathcal E[\sigma_{N,m,\epsilon,0}+U]
 =E_{N,m}+\frac12\int_M
 \bigl[\lambda_{\rm el}(\operatorname{div}U)^2+2\mu|e(U)|^2\bigr]dV_g.
\]
As in Theorem~\ref{thm:halfturn_equilibrium}, equality holds exactly
for $U=ce_z$. If a section in the class is merely assumed stationary,
its first variation in direction $U$ makes the same nonnegative
quadratic integral vanish. It is therefore one of these minimizers.
\end{proof}

The torsion labels $\epsilon_x,\epsilon_y$ do not enter the energy.
Distinct winding labels can also have the same symmetric strain;
uniqueness is asserted only within a fixed class, modulo axial shifts.
These statements concern the specified quadratic phase model. Large
winding gradients need not lie in a physical small-displacement-gradient
regime; the all-sector classification is not a finite-strain theory.

\begin{corollary}[A positive lower bound over all sectors]
\label{cor:global_energy_bound}
Under \eqref{eq:halfturn_positivity}, define
\begin{equation}
 Y_{\rm el}=\frac{\mu(3\lambda_{\rm el}+2\mu)}{\lambda_{\rm el}+\mu}>0.
 \label{eq:young_modulus_bound}
\end{equation}
Every smooth section of this nonsymmorphic bundle satisfies
\begin{equation}
 \mathcal E[\sigma]\ge
 \frac{L_xL_y\ell_z^2}{8L}\,Y_{\rm el}>0.
 \label{eq:all_sector_lower_bound}
\end{equation}
This bound is attained if and only if there exist
$m\in\{0,1\}$ and $n_{11},n_{22}\in\mathbb Z$ such that
\begin{equation}
 \frac{\ell_xn_{11}}{L_x}
 =\frac{\ell_yn_{22}}{L_y}
 =-\frac{\lambda_{\rm el}}{2(\lambda_{\rm el}+\mu)}\,\gamma_m.
 \label{eq:bound_attainment_condition}
\end{equation}
For an attaining section its remaining winding numbers satisfy $h_N=0$;
the torsion labels and the axial constant are arbitrary.
\end{corollary}
\begin{proof}
The positivity assumptions imply $\lambda_{\rm el}+\mu>\mu/3>0$.
For a symmetric $3\times3$ tensor $e$ with $e_{zz}=q$, direct completion
of squares gives
\[
\begin{aligned}
 W(e)-\tfrac12Y_{\rm el}q^2
 &=\frac\mu2(e_{xx}-e_{yy})^2
 +\frac{\lambda_{\rm el}+\mu}{2}
 \left(e_{xx}+e_{yy}+\frac{\lambda_{\rm el}}{\lambda_{\rm el}+\mu}q\right)^2\\
 &\quad+2\mu(e_{xy}^2+e_{xz}^2+e_{yz}^2).
\end{aligned}
\]
Thus minimizing over the remaining entries gives
\[
 e_{xx}=e_{yy}=-\frac{\lambda_{\rm el}}{2(\lambda_{\rm el}+\mu)}q,
 \qquad e_{ij}=0\ (i\ne j),\qquad
 W(e)\ge\tfrac12Y_{\rm el}q^2.
\]
Strict convexity makes these equality conditions necessary and sufficient.
For a section in sector $(N,m,\epsilon)$, write
$\sigma=\sigma_{N,m,\epsilon,0}+U$. Since $e_z$ is globally parallel,
$e(U)_{zz}=e_z(g(U,e_z))$ has integral zero on $M$. Hence
\[
 \frac1{\operatorname{Vol}(M)}\int_M\varepsilon(\sigma)_{zz}\,dV_g
 =\gamma_m,\qquad |\gamma_m|\ge\frac{\ell_z}{2L}.
\]
The pointwise inequality and the mean-square inequality give
\[
 \mathcal E[\sigma]\ge\tfrac12Y_{\rm el}
     \int_M\varepsilon(\sigma)_{zz}^2dV_g
 \ge\tfrac12Y_{\rm el}\operatorname{Vol}(M)\gamma_m^2,
\]
proving \eqref{eq:all_sector_lower_bound}. Equality forces
$|\gamma_m|=\ell_z/(2L)$ and the displayed constant minimizing strain.
Such a section is a minimizer in its own sector, hence is an affine
representative by Theorem~\ref{thm:all_sector_minima}. Its strain gives
\eqref{eq:bound_attainment_condition} and $h_N=0$. Conversely these
conditions, with $n_{12}=n_{21}=0$ if necessary, yield an attaining
representative.
\end{proof}

The actual global minimum is attained even when
\eqref{eq:bound_attainment_condition} fails. More explicitly, let
$d_x=\ell_x/L_x$, $d_y=\ell_y/L_y$, $d_z=\ell_z/L$. Then
\begin{equation}
\begin{split}
 \min_{\sigma}\mathcal E[\sigma]
 =L_xL_yL\min_{j,k,m\in\mathbb Z}\Bigl\{
 &\frac{\lambda_{\rm el}}2
      \bigl(d_xj+d_yk+d_z(m-\tfrac12)\bigr)^2\\
 &+\mu\bigl[(d_xj)^2+(d_yk)^2+d_z^2(m-\tfrac12)^2\bigr]
 \Bigr\}.
\end{split}
 \label{eq:global_discrete_minimum}
\end{equation}
Indeed, the shear term in \eqref{eq:all_sector_energy} is nonnegative
and can be set to zero independently by $n_{12}=n_{21}=0$.
The remaining positive-definite quadratic form has finite sublevel sets
on the indicated shifted rectangular lattice, so its minimum exists.
Thus \eqref{eq:all_sector_lower_bound} is a uniform bound with an explicit
attainment criterion, not an unconditional formula for that minimum.
Nor is $\sigma_0$ asserted to minimize across all sectors. Both this bound
and the discrete minimization problem \eqref{eq:global_discrete_minimum} are
statements about the specified quadratic phase model. Sectors with large
winding gradients can lie outside the physical small-displacement-gradient
regime, so these all-sector results should not be interpreted as a nonlinear
elastic theorem for an arbitrary crystal.

\subsection{Self-adjoint fluctuation operator}

With the normalization of \eqref{eq:elastic_second_variation},
Theorem~\ref{thm:global_hessian} gives the following quadratic fluctuation
action on a finite time interval $I$ for global vector fluctuations $U$
about any stationary section constructed above:
\begin{equation}
 S^{(2)}[U]=\frac12\int_{I\times M}
 \bigl[\rho_m|\dot U|^2-\lambda_{\rm el}(\operatorname{div}U)^2
                          -2\mu|e(U)|^2\bigr]dt\,dV_g.
 \label{eq:halfturn_fluctuation_action}
\end{equation}
Thus $\partial_t^2U+\mathcal AU=0$, where
\begin{equation}
 \mathcal A=\rho_m^{-1}
 \bigl[\mu\nabla^*\nabla
       -(\lambda_{\rm el}+\mu)\operatorname{grad}\operatorname{div}\bigr].
 \label{eq:halfturn_operator}
\end{equation}
Here $\nabla^*\nabla$ is the nonnegative rough Laplacian. On the covering
space the operator is
$-\rho_m^{-1}[\mu\Delta+(\lambda_{\rm el}+\mu)\nabla\operatorname{div}]$.
Its symbol is symmetric positive definite away from zero. We use the
unweighted space $L^2(M,TM;dV_g)$. The corresponding bilinear form is
\begin{equation}
 q_{\mathcal A}(U,V)=\frac1{\rho_m}\int_M
 \bigl[\lambda_{\rm el}(\operatorname{div}U)(\operatorname{div}V)
                 +2\mu\,e(U):e(V)\bigr]dV_g,
 \quad \mathcal D(q_{\mathcal A})=H^1(M,TM).
 \label{eq:mass_normalized_form}
\end{equation}
This is the static Hessian divided by $\rho_m$, not the static energy
itself. Positivity and the identity in the proof of
Theorem~\ref{thm:halfturn_equilibrium} show that its form norm is
equivalent to the $H^1$ norm. It is therefore a closed nonnegative form.
Its self-adjoint operator has domain $H^2(M,TM)$ by elliptic regularity
on the closed manifold, and has compact resolvent. Equivalently, the
unscaled static Hessian defines the same $\mathcal A$ in the
mass-weighted space with inner product
$\langle U,V\rangle_{\rho_m}=\rho_m\int_M g(U,V)\,dV_g$. The proof of Theorem~\ref{thm:halfturn_equilibrium} gives
\begin{equation}
 \ker\mathcal A=\mathbb R e_z.
 \label{eq:halfturn_kernel}
\end{equation}
On a cut fundamental region, the domain requires all matching conditions
inherited from a global $H^2$ vector field, including derivative traces;
matching displacement values alone would not specify this operator.

\begin{proposition}[Hodge form of the isotropic operator]
\label{prop:hodge_elasticity}
On a closed compatible flat spatial manifold with constant isotropic
coefficients satisfying $c_T^2,c_L^2>0$, identify vector fields with
one-forms by $\flat$.
Let $d$ be exterior differentiation, $\delta$ its $L^2$ adjoint and
$\Delta_p=d\delta+\delta d$ the nonnegative Hodge Laplacian on $p$-forms.
Then
\begin{equation}
 \mathcal A^\flat:=\flat\circ\mathcal A\circ\sharp
       =c_T^2\delta d+c_L^2d\delta.
 \label{eq:hodge_elastic_operator}
\end{equation}
Its restrictions to exact and coexact one-forms are, respectively,
$c_L^2\Delta_1$ and $c_T^2\Delta_1$. Its kernel consists of harmonic
one-forms, which are parallel in the flat case, and
\begin{equation}
 \dim\ker\mathcal A=b_1(M).
 \label{eq:hodge_elastic_kernel}
\end{equation}
If $m_p(\eta)=\dim\ker(\Delta_p-\eta)$ for $\eta>0$, with value zero
off the spectrum, the multiplicity of $\zeta>0$ for $\mathcal A$ is
\begin{equation}
 m_{\mathcal A}(\zeta)=
 m_0(\zeta/c_L^2)+m_1(\zeta/c_T^2)-m_0(\zeta/c_T^2).
 \label{eq:elastic_hodge_multiplicity}
\end{equation}
\end{proposition}
\begin{proof}
The flat Weitzenb\"ock identity gives $\Delta_1=\nabla^*\nabla$, while
$\delta(U^\flat)=-\operatorname{div}U$. Substitution in
\eqref{eq:halfturn_operator} gives \eqref{eq:hodge_elastic_operator}.
The Hodge decomposition
\[
 L^2\Omega^1(M)=\mathcal H^1(M)\oplus
     \overline{d\Omega^0(M)}\oplus\overline{\delta\Omega^2(M)}
\]
is orthogonal and invariant under the operator. Positivity of $c_T^2$
and $c_L^2$ identifies its kernel with $\mathcal H^1(M)$; the Hodge
theorem identifies this space with $H^1(M;\mathbb R)$
\cite{RosenbergLaplacian}. For a harmonic one-form $\theta$, flatness gives
$0=\langle\Delta_1\theta,\theta\rangle=\|\nabla\theta\|^2$, proving
parallelness. On each positive $\Delta_1$ eigenspace the exact part has
dimension $m_0(\eta)$, since $d$ identifies the scalar eigenspace with
that exact part. Its orthogonal coexact complement has dimension
$m_1(\eta)-m_0(\eta)$. Scaling the two parts proves
\eqref{eq:elastic_hodge_multiplicity}, including accidental coincidences.
\end{proof}

This is an application of standard Hodge theory, not a new spectral
construction. It relates the elastic calculation to the equivariant
torus methods and multiplicity formulas for flat manifolds and natural
vector bundles in \cite{MiatelloRossettiForms,MiatelloRossettiTwisted}.
In the present example the holonomy-invariant constant one-forms are
precisely the multiples of $dz$, recovering \eqref{eq:halfturn_kernel}.
Here $b_1(M)=\dim H^1(M;\mathbb R)=1$ uses ordinary real coefficients;
it must not be confused with the twisted lattice cohomology
\eqref{eq:halfturn_h1} classifying configurations.
The next subsection makes all corresponding vector modes explicit.

\subsection{Complete mode decomposition and the first positive eigenvalue}

A lifted vector field satisfies
\begin{equation}
 U(T_xx)=U(x),\qquad U(T_yx)=U(x),\qquad U(Sx)=R U(x).
 \label{eq:halfturn_vector_equivariance}
\end{equation}
The Fourier wave vectors on $\widehat M$ are
\begin{equation}
 k_{pqr}=\left(\frac{2\pi p}{L_x},\frac{2\pi q}{L_y},\frac{\pi r}{L}\right),
 \qquad p,q,r\in\mathbb Z.
 \label{eq:halfturn_wave_vectors}
\end{equation}
For Fourier coefficients $a_k$, condition
\eqref{eq:halfturn_vector_equivariance} is exactly
\begin{equation}
 a_{Rk}=(-1)^rR a_k.
 \label{eq:halfturn_fourier_constraint}
\end{equation}
Denote the speeds in \eqref{eq:isotropic_speeds} by $c_L,c_T$.

\begin{theorem}[Complete isotropic half-turn spectrum]
\label{thm:halfturn_spectrum}
The following complex eigenfields span the complexified $L^2$ space of
vector fields on $M$.
\begin{enumerate}[label=\textup{(\roman*)}]
\item For $(p,q)\ne(0,0)$ choose one representative of
$(p,q,r)\sim(-p,-q,r)$, for instance $p>0$ or $p=0,q>0$. For $k=k_{pqr}$,
\begin{equation}
 \Phi_{k,a}(x)=a e^{ik\cdot x}+(-1)^rRa e^{iRk\cdot x}.
 \label{eq:halfturn_nonaxial_modes}
\end{equation}
If $a$ is parallel to $k$, its eigenvalue is $c_L^2|k|^2$.
If $a\cdot k=0$, its eigenvalue is $c_T^2|k|^2$, with two independent
polarizations per orbit.
\item For $p=q=0$ the longitudinal modes and eigenvalues are
\begin{equation}
 U_n=e_z e^{2\pi inz/L},\qquad
 \omega_{L,n}^2=c_L^2(2\pi n/L)^2,\qquad n\in\mathbb Z.
 \label{eq:halfturn_axial_longitudinal}
\end{equation}
The case $n=0$ is the single constant polarization $e_z$.
\item The axial transverse modes and eigenvalues are
\begin{equation}
 V_{n,a}=a e^{i(2n+1)\pi z/L},\qquad
 \omega_{T,n}^2=c_T^2((2n+1)\pi/L)^2,
 \quad a\in\operatorname{span}_{\mathbb C}\{e_x,e_y\}.
 \label{eq:halfturn_axial_transverse}
\end{equation}
Here $n\in\mathbb Z$. There are no other axial polarizations.
\end{enumerate}
The first positive eigenvalue is
\begin{equation}
 \lambda_1^+(\mathcal A)=c_T^2
 \min\left\{\frac{4\pi^2}{L_x^2},\frac{4\pi^2}{L_y^2},
                              \frac{\pi^2}{L^2}\right\}.
 \label{eq:halfturn_spectral_gap}
\end{equation}
\end{theorem}
\begin{proof}
Fourier expansion on the covering torus is complete.
For $U=\sum_k a_ke^{ik\cdot x}$, the identity $U(Sx)=RU(x)$ gives
$(-1)^r a_{Rk}=Ra_k$, proving \eqref{eq:halfturn_fourier_constraint}.
For a two-element orbit it leaves $a_k$ arbitrary and determines $a_{Rk}$,
which gives \eqref{eq:halfturn_nonaxial_modes}. The symbol
\[
 \widehat{\mathcal A}(k)
 =\rho_m^{-1}\bigl[\mu|k|^2I
                     +(\lambda_{\rm el}+\mu)k\otimes k\bigr]
\]
has eigenvalues $c_L^2|k|^2$ and $c_T^2|k|^2$ on the longitudinal and
transverse spaces, respectively, and obeys
$\widehat{\mathcal A}(Rk)=R\widehat{\mathcal A}(k)R^{-1}$.
Thus both terms of \eqref{eq:halfturn_nonaxial_modes} have the same
eigenvalue.

For a one-element orbit, $k=(0,0,\pi r/L)$, the constraint is
$a=(-1)^rRa$. If $r$ is even, $a$ is axial; if $r$ is odd, $a$ is
transverse. This proves the axial formulas, including the kernel.
All allowed Fourier coefficients have been exhausted, so the fields are
complete. Within each eigenspace they can be normalized to an orthonormal
basis; coinciding eigenvalues from different orbits have their
multiplicities added. Real fields are obtained with the usual complex
conjugacy condition.

Finally, \eqref{eq:halfturn_positivity} gives
$\lambda_{\rm el}+\mu>\mu/3>0$, hence $c_L>c_T$.
The smallest nonaxial transverse values occur at $r=0$ with
$(p,q)=(1,0)$ or $(0,1)$, and the smallest axial transverse value has
$|r|=1$. The smallest nonzero axial longitudinal value is
$c_L^2(2\pi/L)^2$, which is larger than $c_T^2(\pi/L)^2$.
These observations prove \eqref{eq:halfturn_spectral_gap}.
\end{proof}

For $L_x=L_y=L$, the lowest positive frequency is $\pi c_T/L$ and its
real eigenspace has dimension four, spanned by
\[
 e_x\cos(\pi z/L),\quad e_x\sin(\pi z/L),\quad
 e_y\cos(\pi z/L),\quad e_y\sin(\pi z/L).
\]
The frequency zero has a one-dimensional eigenspace. These are finite-volume
features; the local elastic speeds remain \eqref{eq:isotropic_speeds}.
The method uses the established equivariant Fourier construction on flat
quotients \cite{MiatelloRossettiForms,MiatelloRossettiTwisted}. The scalar
half-turn modes are also discussed in \cite{EskiltHalfTurn}; here the
vector constraint and the elastic longitudinal/transverse decomposition
are imposed together.

\subsection{What depends on the affine extension}

Keep $M,Q,\Pi$ and the material coefficients fixed, but replace
\eqref{eq:halfturn_affine_action} by the linear action $[u]\mapsto[Ru]$.
This corresponds to $\Gamma_{\rm sym}=\Pi\rtimes C_2$ and admits the
parallel zero section of zero energy. In the nonsymmorphic model no
parallel section exists, and every smooth section has the strictly
positive lower energy bound \eqref{eq:all_sector_lower_bound}. In particular
the reference sector has energy \eqref{eq:halfturn_background_energy}.
The associated vector bundle,
its connection and the Hessian are nevertheless identical in the two
models. Hence their spectra about these stationary references coincide.
This separation follows from phase independence and quadraticity, not
from affine gluing alone. A phase-dependent or nonlinear constitutive law
may have a Hessian that depends on the forced background strain.

The comparison changes the microscopic extension $\Gamma$ only; it
keeps the macroscopic deck group $\mathcal D$, the flat metric and the operator
domain fixed. It is not independence from the translational parts of
$\mathcal D$: the spatial shift $z\mapsto z+L$ produces the phase $(-1)^r$ in
\eqref{eq:halfturn_fourier_constraint}, and the macroscopic lengths enter
the eigenvalues. The constitutive law is postulated at the continuum
level, not derived from a microscopic lattice dynamics or a specific
material.

The axial formulas also specify the precise setting in which a shifted
wave-number rule is spectral information: the operator is homogeneous,
compatible with the holonomy, and its complete domain is prescribed.
The holonomy and length of an arbitrary loop by themselves do not determine
the spectrum of a multidimensional elastic problem.

\section{Conclusion}
\label{sec:conclusion}

The crystallographic extension determines an affine action on the translation
torus and hence the configuration bundle $Y_\Pi=Q\times_P T_\Pi$. The first
main conclusion is that global existence and the existence of a parallel
reference are genuinely different questions. Smooth configurations are
controlled by $f_Q^*[c]\in H^2(X;\Pi_Q)$, whereas parallel configurations
are controlled by $[c]|_{H_Q}$. The half-turn model realizes the strict
separation
\[
 f_Q^*[c]=0,\qquad [c]|_{H_Q}\ne0,
\]
on a compatible three-dimensional spatial manifold: a global translational
configuration exists although a parallel reference is topologically
forbidden.

The second conclusion is variational. The affine torsor identity
$\beta_{\sigma_0+\varphi}=\beta_{\sigma_0}+D^E\varphi$ globalizes the
linear fluctuation variable even about a nonparallel background. For a
phase-independent quadratic first-order theory it yields an exact expansion
and a background-independent Hessian, so stationarity is sufficient for the
linearized equations. In the half-turn model the homotopy classes form a
torsor for $\mathbb Z^5\oplus(\mathbb Z/2\mathbb Z)^2$; every sector has an
explicit minimizing equilibrium, and the full configuration space of the
quadratic model obeys a strictly positive lower energy bound. Thus the
failure of a parallel reference is not merely classificatory in this
example: together with the prescribed global geometry and constitutive law,
it has a concrete energetic consequence.

The third conclusion concerns what linear fluctuations do and do not retain
from the affine extension. The complete isotropic spectrum is given in
Theorem~\ref{thm:halfturn_spectrum}; the half-turn holonomy permits even
axial longitudinal and odd axial transverse harmonics and leaves only one
constant zero-mode polarization. Its Hodge form
$\mathcal A^\flat=c_T^2\delta d+c_L^2d\delta$ explains the kernel and the
longitudinal--transverse decomposition. At the same time, when the
macroscopic quotient, flat linear bundle, and constitutive coefficients are
held fixed, replacing the nonsymmorphic affine action by its symmorphic
linear counterpart changes the allowed configuration space and sector
energies but not the quadratic fluctuation operator. Configuration topology
and linear fluctuation physics therefore separate in this
phase-independent quadratic limit.

The standard elastic interpretation still requires compatibility of $Q$ with
Levi--Civita parallel transport. Without it, the induced flat metric
connection may have torsion and its symmetrized derivative need not be the
ordinary small strain of a physical displacement. Moreover, the positive
all-sector energy bound is a result of the specified quadratic phase model;
large winding sectors need not remain within the small-gradient regime. The
spectral methods and general affine-torus obstruction theory used here are
established tools; the contribution is the global relation between the
crystallographic extension, the two obstruction problems, sectorwise elastic
energetics, and the vector-valued spectral problem.

The orientational reduction remains fixed. Nonlinear strain,
phase-dependent coefficients, dynamical orientational order, or dynamical
defects require additional compatibility and constitutive input. In such
extensions the affine data can re-enter the Hessian through the forced
background, so the separation established here need not persist.

\appendix
\section{Cubic Elasticity on a Compatible Background}
\label{app:cubic_example}

Let $\dim M=3$ and let $Q\subset F_g(M)$ be a compatible reduction to
$P=O_h$ or its orientation-preserving subgroup. On
$X=\mathbb R_t\times M$ use its time-independent pullback. Choose a
$P$-invariant lattice $\Pi$ and a space-group extension with this linear
action. Whenever a stationary reference section exists, its global
fluctuations lie in $E_\Pi\simeq\mathrm{pr}_M^*TM$. The quadratic
phase-independent Hessian does not require this reference to be parallel.

Writing the fluctuation as $U$, its velocity and strain are
\[
 v=\partial_tU,\qquad
 e=\operatorname{sym}((\nabla^gU)^\flat).
\]
In local parallel Cartesian cubic frames,
$e_{ij}=\frac12(\partial_iU_j+\partial_jU_i)$.
The cubic energy is the transported $P$-invariant form
\begin{equation}
\begin{split}
 W_{\rm cub}(e)={}&\frac12C_{11}(e_{11}^2+e_{22}^2+e_{33}^2)\\
 &+C_{12}(e_{11}e_{22}+e_{22}e_{33}+e_{33}e_{11})\\
 &+2C_{44}(e_{12}^2+e_{23}^2+e_{31}^2).
\end{split}
\label{eq:cubic_energy}
\end{equation}
The Lagrangian density is $\rho_m|\dot U|^2/2-W_{\rm cub}(e)$.
With stress defined by $\delta W=\Sigma_{ij}\delta e_{ij}$, summing over
all ordered pairs $i,j$, one obtains
\[
 \Sigma_{ii}=C_{11}e_{ii}+C_{12}\sum_{j\ne i}e_{jj},\qquad
 \Sigma_{ij}=2C_{44}e_{ij}\quad(i\ne j).
\]
Thus $\rho_m\ddot U_i=\partial_j\Sigma_{ij}$.
For constant material coefficients, the local Christoffel matrix is
\begin{equation}
 \mathcal C_{ij}(k)=C_{44}|k|^2\delta_{ij}
 +(C_{12}+C_{44})k_ik_j
 +(C_{11}-C_{12}-2C_{44})k_i^2\delta_{ij},
 \label{eq:cubic_christoffel}
\end{equation}
with no summation over $i$ in the last term
\cite{NyePhysicalProperties,MusgraveCrystalAcoustics}. In the $[100]$
direction this gives
\[
 \omega_L^2=(C_{11}/\rho_m)|k|^2,\qquad
 \omega_T^2=(C_{44}/\rho_m)|k|^2,
\]
and in the $[111]$ direction
\[
 \omega_L^2=\frac{C_{11}+2C_{12}+4C_{44}}{3\rho_m}|k|^2,
 \qquad
 \omega_T^2=\frac{C_{11}-C_{12}+C_{44}}{3\rho_m}|k|^2.
\]
The transverse branch is doubly degenerate in each of these directions.
Positive strain energy requires
$C_{44}>0$, $C_{11}-C_{12}>0$ and $C_{11}+2C_{12}>0$; also $\rho_m>0$.
These are local acoustic formulas. A global finite-volume spectrum needs
the spatial identifications and vector-field domain in addition to the
material symmetry.


\end{document}